\RequirePackage{ifpdf}
\ifpdf 
\documentclass[pdftex]{sigma}
\else
\documentclass{sigma}
\fi

\usepackage[mathscr]{euscript}
\usepackage{mathrsfs,bm}

\numberwithin{equation}{section}

\newcommand{\one}{{\mathchoice {\rm 1\mskip-4mu l} {\rm 1\mskip-4mu l}
{\rm 1\mskip-4.5mu l} {\rm 1\mskip-5mu l}}}
\newcommand{\h}{\mathfrak{h}}

\newcommand{\ex}{\mathrm{e}}
\newcommand{\D}{\mathrm{dom}}
\newcommand{\R}{\mathrm{ran}}
\newcommand{\im}{\mathrm{i}}
\newcommand{\Fock}{\mathfrak{F}}

\newcommand{\dG}{\mathrm{d}\Gamma}

\newcommand{\la}{\langle}
\newcommand{\ra}{\rangle}
\newcommand{\Tr}{\operatorname{tr}}

\newcommand{\BbbR}{\mathbb{R}}
\newcommand{\BbbN}{\mathbb{N}}
\newcommand{\BbbZ}{\mathbb{Z}}

\newcommand{\BbbC}{\mathbb{C}}
\newcommand{\vepsilon}{\varepsilon}
\newcommand{\vphi}{\varphi}

\newcommand{\Nf}{N_{\mathrm{f}}}
\newcommand{\Hf}{H_{\mathrm{f}}}

\newcommand{\A}{\mathscr{A}}

\newcommand{\dm}{\mathrm{d}}

\newcommand{\no}{\nonumber \\}

\newcommand{\bs }{\boldsymbol}

{ \theoremstyle{definition}
\newtheorem{define}{Definition}[section]
\newtheorem{rem}[define]{Remark}

\newtheorem{exam}[define]{Example}
}

\newtheorem{Thm}[define]{Theorem}
\newtheorem{Prop}[define]{Proposition}
\newtheorem{lemm}[define]{Lemma}
\newtheorem{coro}[define]{Corollary}

\begin{document}

\allowdisplaybreaks

\newcommand{\arXivNumber}{1902.05207}

\renewcommand{\PaperNumber}{036}

\FirstPageHeading

\ShortArticleName{Note on the Retarded van der Waals Potential within the Dipole Approximation}

\ArticleName{Note on the Retarded van der Waals Potential\\ within the Dipole Approximation}

\Author{Tadahiro MIYAO}

\AuthorNameForHeading{T.~Miyao}

\Address{Department of Mathematics, Hokkaido University, Sapporo, Japan}
\Email{\href{mailto:miyao@math.sci.hokudai.ac.jp}{miyao@math.sci.hokudai.ac.jp}}

\ArticleDates{Received February 27, 2019, in final form April 14, 2020; Published online April 26, 2020}

\Abstract{We examine the dipole approximated Pauli--Fierz Hamiltonians of the nonrelativistic QED. We assume that the Coulomb potential of the nuclei together with the Coulomb interaction between the electrons can be approximated by harmonic potentials. By an exact diagonalization method, we prove that the binding energy of the two hydrogen atoms behaves as $R^{-7}$, provided that the distance between atoms $R$ is sufficiently large. We employ the Feynman's representation of the quantized radiation fields which enables us to diagonalize Hamiltonians, rigorously. Our result supports the famous conjecture by Casimir and Polder.}


\Classification{81V10; 81V55; 47A75}

\section{Introduction}

London was the first to explain attractive interactions between neutral atoms or molecules by applying quantum mechanics~\cite{London}. Nowadays, the attractive forces are called the van der Waals--London forces, and are described by the potential energy decaying as~$R^{-6}$ for $R$ sufficiently large.\footnote{More precisely, if one takes the interactions between electrons and the quantized Maxwell field according to non-relativistic QED into account, the $R^{-6}$ behavior is true for the near-field region (very vaguely ``sufficiently large $R$ but not too large'', and discussed on \cite[p.~157]{Craig} and \cite{MS2}), but for the far-field region (where ``retardation effects become important'') the presented results show a $R^{-7}$ behavior. In the approximation where the quantum fluctuations of the Maxwell field are ignored, only the electrostatic Coulomb interaction remains. In this case, the binding energy behaves as $R^{-6}$ provided that $R$ is sufficently large. This $R^{-6}$ behavior is well-understood, mathematically \cite{AMR,AS, LT, MoSi}.} Here, $R$ denotes the distance between two atoms or molecules. It is recognized that these forces come from the quantum fluctuations of the charges inside the atoms. Because even a simple hydrogen atom displays a fluctuating dipole, the van der Waals--London forces are ubiquitous and therefore very fundamental.

Casimir and Polder took the interactions between electrons and the quantized radiation fields
into consideration and perfomed the fourth order perturbative computations~\cite{CP}. They found that
the finiteness of the speed of light weakens the correlation between nearby dipoles and causes
the attractive potential between atoms to behave as
\begin{gather}
V_{\mathrm{CP}}(R)\cong -\frac{23}{4\pi} \bigg(\frac{1}{2\pi}\bigg)^2\frac{1}{R^7} \alpha_A \alpha_B,\qquad R\gg 1, \label{CPpo}
\end{gather}
where $\alpha_A$ and $\alpha_B$ are the static polarizability of the atoms.
The potential $V_{\mathrm{CP}}$ is called the Casimir--Polder potential or the retarded van der Waals potential. For reviews, see, e.g., \cite{BMM,Keller, Ex, MB,Milonni}.
Although this result is plausible, Casimir--Polder's arguments are heuristic, and lack mathematical rigor.

There are few rigorous results concerning the Casimir--Polder potential; In \cite{MS1, MS2}, Miyao and Spohn gave a path integral formula for $V_{\mathrm{CP}}$ and applied it to computing the second cumulant.
Under the assumption that all of higher order cumulants behave as $O\big(R^{-9}\big)$ and their
coefficients are small enough to control, they rigorously refound that $V_{\mathrm{CP}}$ behaves as $R^{-7}$ as $R\to \infty$. Although this assumption appears to be plausible, to prove it is extremely hard.
Therefore, to give a mathematical foundation of the Casimir--Polder potential is an open problem even today.

In the present paper, we will examine the Pauli--Fierz model under the following assumptions \cite[equations~(13.127) and~(13.123)]{Spohn}:
\begin{itemize}\itemsep=0pt
\item[(C.1)] the dipole approxiamtion (see~(\ref{DIPA}));
\item[(C.2)] the electrons are strongly bound around each nucleus (see (\ref{BOUND1}) and (\ref{BOUND2})).
\end{itemize}
The dipole approximation (C.1) is widely accepted as a convenient procedure
in the community of the nonrelativistic QED~\cite{Spohn}. The assumption (C.2) is often useful when we study the low energy behavior of the system. Under the assumptions, we prove that the binding energy for two hydrogen atoms actually behaves as $R^{-7}$. In the context of the Born--Oppenheimer approximation, this indicates that the effective potential between two hydrogen atoms behaves as $R^{-7}$ too.
This result supports our assumptions for the model without dipole
approximation, and is expected to become a starting point for study of the non-approximated model.
 Our proof relies on the fact that the dipole approximated Hamiltoninas can be diagonalized
by applying Feynman's representation of the quantized radiation fields~\cite{Feynman}.
It has been believed that the dipole approximated model also exhibits~$R^{-7}$ behavior by the forth order perturbation theory. However, the arguments concering the error terms are completely missing. Indeed, this part is tacitly assumed to be trivial in literatures. In this paper, we actually perform systematic error estimates which are far from trivial.

In mathematical physics, it is known that rigorous studies of the Pauli--Fierz Hamiltonian require an extra care due to the infamous infrared problem \cite{BFS, GLL, Spohn}.
Fortunetely, within the assumptions (C.1) and (C.2), we can control the problem relatively easily.

Before we proceed, we have additional remarks.
In his Ph.D.~Thesis \cite{Koppen}, Koppen studied the retarded van der Waals potential; he examined the Pauli--Fierz model with the dipole approximation~(C.1), but the condition~(C.2) is not assumed in \cite{Koppen}. In contrast to the present study, he imposed the infrared cutoff $\sigma$ on the Hamiltonian in order to apply the naive perturbation theory and obtained an expansion formula for the binding energy: $E_{\sigma}(R)=\sum\limits_{i=0}^{\infty} e^i V_i^{\sigma}(R)$. Then he removed the infrared cutoff from each term: $V_i(R):=\lim\limits_{\sigma\to +0} V_i^{\sigma}(R)$. Finally, he proved that some $V_i(R)$ satisfies~(\ref{CPpo}). His observation could be regareded as a nice starting point of mathematical analysis of the retarded van der Waals potential, however, there are still some problems to be considered. For example, the magnetic contributions to the~$-R^{-7}$ decay are completely overlooked.
 In addition, in the mathematical study of the Pauli--Fierz model, it is well-known that to prove
that $\lim\limits_{\sigma\to +0}E_{\sigma}(R)=\sum\limits_{i=0}^{\infty} e^i \lim\limits_{\sigma\to +0}V_i^{\sigma}(R)$ is very hard problem, the aforementioned infrared problem.

Our contributions are
\begin{itemize}\itemsep=0pt
\item to provide a minimal QED model which can rigorously explain the Casimir--Polder potential by a relatively simple and easy way;
\item to perform systematic error estimates without the infrared cutoff.
\end{itemize}
In this way, the present paper and the thesis \cite{Koppen} are complementary to each other.

Since the electrons obey Fermi--Dirac statistics, the wave functions of the two-electron system belong to
$(\mathfrak{H}\wedge \mathfrak{H})\otimes \Fock\big(L^2\big(\BbbR^3\times \{1, 2\}\big)\big)$,
where $\mathfrak{H}=L^2\big(\BbbR^3\big)\otimes \BbbC^2$, the Hilbert space with spin $1/2$,
the symbol $\wedge $ indicates the anti-symmetric tensor product and $\Fock\big(L^2\big(\BbbR^3\times \{1, 2\}\big)\big)$ is the Fock space over $L^2\big(\BbbR^3\times \{1, 2\}\big)$.
Usually, the ground state of this system is a spin singlet. Thus, the spatial part of the ground state is symmetric and we can end up with minimizing the energy in an unrestricted manner on $\big(L^2\big(\BbbR^3\big)\otimes L^2\big(\BbbR^3\big)\big)\otimes \Fock\big(L^2\big(\BbbR^3\times \{1, 2\}\big)\big)$. For this reason, we perform our analysis on $\big(L^2\big(\BbbR^3\big)\otimes L^2\big(\BbbR^3\big)\big)\otimes \Fock\big(L^2\big(\BbbR^3\times \{1, 2\}\big)\big)$.\footnote{Or we could simply say that one considers the ``distinguishable particles'', see Section~\ref{Discuss} for detail.}
However, it should be mentioned that our observation here can not be extended to general $N$-electron systems, directly.

In fairness, we mention the following two difficulties of the assumptions~(C.1) and~(C.2). For details, see discussions in Section~\ref{Discuss}.
\begin{itemize}\itemsep=0pt
\item The condition (C.2) breaks the indistinguishability of the electrons.
\item Under the conditions (C.1) and (C.2), we cannot reproduce the exact cancellation of the term with $R^{-6}$ decay (the van der Waals--London potential) by the contribution from the quantized Maxwell field. Note that this cancellation is known to be fundamental to explain the retarded van der Waals potential \cite{MS1, MS2}.
\end{itemize}

The present paper is organized as follows. In Section~\ref{SecM}, we introduce the dipole approximated Pauli--Fierz Hamiltonian and state the main result. In Section~\ref{FeynRep}, we switch to the Feynman representation of the quantized radiation fields. This representation enables us to diagonalize the Hamiltonians as we will see in the following sections. Further, we introduce a canonical transformation which induces the quantized displacement fields in the Hamiltonians in Section~\ref{CanoTr}. Section~\ref{FiniteVApp} is devoted to the finite volume approximation, which is a standard method in the
study of the quantum field theory~\cite{AH, GJ}. Then we diagonalize the Hamiltonians in Sections~\ref{Dai1} and~\ref{Dai2}. In Section~\ref{PfMainT}, we give a proof of the main theorem. Section~\ref{Discuss} is devoted to the discussions of the approximations~(C.1) and~(C.2). In Appendices~\ref{List},~\ref{NumC} and~\ref{BasicI}, we collect various auxiliary results which are needed in the main sections.

\section{Main result}\label{SecM}
Let us consider a single hydrogen atom with an infinitely heavy nucleus located at the origin~$0$. The nonrelativistic QED Hamiltonian for this system is given by
\begin{gather*}
H_{\mathrm{1e}}=\frac{1}{2}\big({-}\im \nabla-eA(x)\big)^2-e^2V(x)+\Hf.
\end{gather*}
The nucleus has charge $e>0$, and the electron has charge $-e$. We assume that the charge distribution $ \varrho$ satisfies the following properties:
\begin{itemize}\itemsep=0pt
\item[(A.1)] $\varrho$ is normalized: $\int_{\BbbR^3} \dm x\, \varrho(x)=1$.
\item[(A.2)] $\varrho(x)=\varrho(-x)$. Thus the Fourier transformation $\hat{\varrho}$ is
real.
\item[(A.3)] $\hat{\varrho}$ is rotation
invariant, $\hat{\varrho}(k) = \hat{\varrho}_{\mathrm{rad}}(|k|)$, of rapid decrease and smooth.
\end{itemize}
The smeared Coulomb potential $V$ is given by
 \begin{gather*}
V(x)=\int_{\BbbR^3}\dm k\, \hat{\varrho}(k)^2|k|^{-2}\ex^{-\im k\cdot x}.
\end{gather*}
The photon annihilation operator is denoted by $a(k, \lambda)$. As usual, this operator satisfies the
standard commutation relation:
\begin{gather*}
[a(k, \lambda), a(k', \lambda')^*]=\delta_{\lambda\lambda'}\delta(k-k').
\end{gather*}
The quantized vector potential $A(x)$ is defined by
\begin{gather*}
A(x)=\sum_{\lambda=1,2}\int_{\BbbR^3}\dm k\, \frac{\hat{\varrho}(k)}{\sqrt{2|k|}}\vepsilon(k, \lambda)\big(\ex^{-\im k\cdot x}a(k, \lambda)^*+\ex^{\im k\cdot x}a(k, \lambda)\big),
\end{gather*}
where $\vepsilon(k, \lambda)=(\vepsilon_1(k, \lambda), \vepsilon_2(k, \lambda), \vepsilon_3(k, \lambda)),\, \lambda=1, 2$ are polarization vectors. For concreteness, we choose as
\begin{gather}
\vepsilon(k, 1)=\frac{(k_2, -k_1, 0)}{\sqrt{k_1^2+k_2^2}},\qquad \vepsilon(k, 2)=\frac{k}{|k|}\wedge \vepsilon(k, 1).\label{PolariDef}
\end{gather}
Note that $A(x)$ is essentially self-adjoint. We will denote its closure by the same symbol.
The field energy $\Hf$ is given by
\begin{gather*}
\Hf=\sum_{\lambda=1,2}\int_{\BbbR^3}\dm k\, |k|a(k,
 \lambda)^*a(k,\lambda).
 \end{gather*}
 The operator $H_{\mathrm{1e}}$ acts in the Hilbert space $
L^2\big(\BbbR^3\big)\otimes \Fock\big(L^2\big(\BbbR^3_k\times \{1,2\}\big)\big)$, where
$\Fock(\h)$ is the bosonic Fock space over~$\h$: $\Fock(\h)=\bigoplus\limits_{n=0}^{\infty} \h^{\otimes_{\mathrm{s}}n}$. Here, $\otimes_{\mathrm{s}}$ indicates the symmetric tensor product.

To examine the Casimir--Polder potential, we consider two hydrogen atoms, one located at the origin and the other at $r=(0, 0, R)$ with $R>0$. For computational convenience, we define the position of the second electron relative to $r$, see Fig.~\ref{Figure1}.
\begin{figure}[t]\centering
\includegraphics{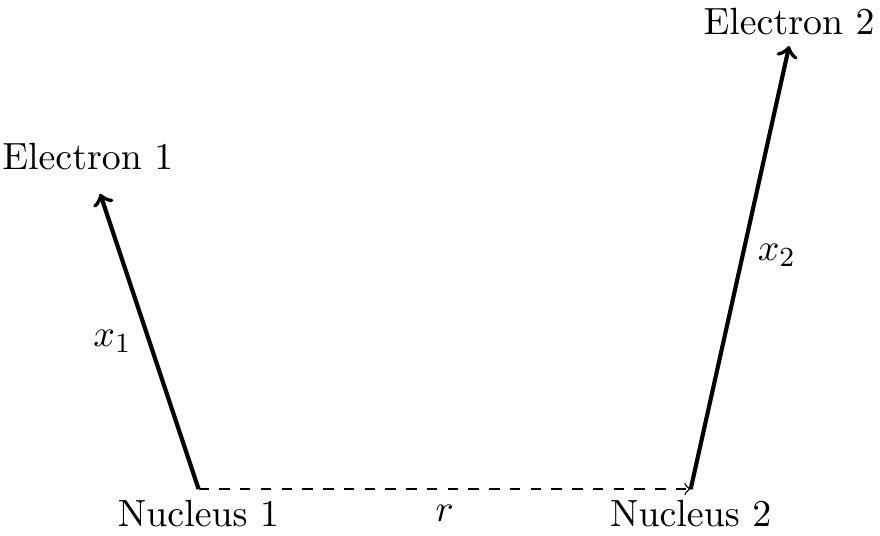}
 \caption{} \label{Figure1}
\end{figure}
Then the two-electron Hamiltonian reads
\begin{gather*}
H_{\mathrm{2e}}=\frac{1}{2}\big({-}\im \nabla_1-eA(x_1)\big)^2-e^2V(x_1)
+\frac{1}{2}\big({-}\im \nabla_2-eA(x_2+r)\big)^2-e^2V(x_2)\\
\hphantom{H_{\mathrm{2e}}=}{}+e^2V_R(x_1, x_2)+\Hf
\end{gather*}
 with
 \begin{align*}
V_R(x_1, x_2)&=-V(x_1-r)-V(x_2+r)+V(r)+V(r+x_2-x_1)\\
&=\int_{\BbbR^3}\dm k\, \hat{\varrho}(k)^2|k|^{-2}\big(1-\ex^{-\im k\cdot
 x_1}\big) \big(1-\ex^{\im k\cdot x_2}\big) \ex^{\im k\cdot r}.
\end{align*}
The operator $H_{\mathrm{2e}}$ acts in $L^2\big(\BbbR^3_{x_1}\big)\otimes L^2\big(\BbbR_{x_2}^3\big)\otimes \Fock\big(L^2\big(\BbbR^3_k\times \{1,2\}\big)\big)$.

The dipole approximation (C.1) means the following replacement:
\begin{gather}
A(x_1)\leadsto A(0),\qquad A(x_2+r)\leadsto A(r). \label{DIPA}
\end{gather}
By the assumption (C.2), we can take $x_1$ and $x_2$ sufficiently small. Therefore, we assume that the Coulomb potential of the nuclei together with the Coulomb interaction between the electrons can be approximated by harmonic potentials. Then one has
\begin{gather}
V(x_j)\simeq -\frac{1}{2}\nu_0^2x_j^2+\mathrm{const} \label{BOUND1}
\end{gather}
 with $\nu^2_0=\frac{1}{3}\int \dm k \, \hat{\varrho}(k)^2$ and
\begin{gather}
V_R(x_1, x_2)\simeq \int\dm k\,
\hat{\varrho}(k)^2 \ex^{\im k\cdot r}\big(x_1\cdot \hat{k}\big) \big(x_2\cdot \hat{k}\big) \label{BOUND2}
\end{gather}
 with $\hat{k}=k/|k|$. Hence, we arrive at
\begin{gather*}
H_{\mathrm{D1e}}=\frac{1}{2}\big({-}\im \nabla -eA(0)\big)^2+\frac{1}{2}e^2\nu_0^2x^2+\Hf
\end{gather*}
and
\begin{gather*}
H_{\mathrm{D2e}} =\frac{1}{2}\big({-}\im
 \nabla_1-eA(0)\big)^2+\frac{1}{2}e^2\nu_0^2 x_1^2
+\frac{1}{2}\big({-}\im \nabla_2-eA(r)\big)^2+\frac{1}{2}e^2\nu_0^2 x_2^2
\\
\hphantom{H_{\mathrm{D2e}} =}{}
 +e^2 \int_{\BbbR^3}\dm k\, \hat{\varrho}(k)^2 \ex^{\im k\cdot r}(x_1\cdot \hat{k}) (x_2\cdot
 \hat{k})+\Hf.
\end{gather*}
 Note that $H_{\mathrm{D1e}}$ and $H_{\mathrm{D2e}}$ are self-adjoint and bounded from below~\cite{LHB}, because the cross-term $\int \dm k\, \hat{\varrho}(k)^2 \ex^{\im k\cdot r} \big(x_1\cdot \hat{k}\big) \big(x_2\cdot \hat{k}\big)$ becomes very small provided that~$R$ is large enough. As for physical discussions of the approximation above, see Section~\ref{Discuss} in detail.

In what follows, we assume an additional condition:
\begin{itemize}\itemsep=0pt
 \item[(A.4)] We regard $\nu_0$ as a parameter. Thus, $\nu_0$ is independent of $\varrho$.
\end{itemize}
Hence, there are three parameters $e$, $R$ and $\nu_0$ in our models.
\begin{Thm}\label{Rto7}
Let $E(R)=\inf \operatorname{spec}(H_{\mathrm{D2e}})$ and let
$E=\inf \operatorname{spec}(H_{\mathrm{D1e}})$, where $\operatorname{spec}(X)$ indicates the spectrum of a linear operator $X$.
Let
\begin{gather*}
c_{\infty}=\max \left\{
\sqrt{2}e \big\||k|^{-1} \hat{\varrho}\big\|_{L^2}, \frac{\||k| \hat{\varrho}\|_{L^2}}{\sqrt{2}e\nu_0^2} ,
\frac{\|\hat{\varrho}\|_{L^2}}{\nu_0} \right\}.
\end{gather*}
Choose $e$ and $\nu_0$ such that $c_{\infty}<1/2$, $1\le \sqrt{2} e \nu_0$ and $\sqrt{2} e \|\hat{\varrho}\|_{L^2}<1$. Then
one has
\begin{gather*}
\lim_{R\to \infty}R^7(2E-E(R))=\frac{23}{4\pi}\left(\frac{1}{2\pi}\right)^2\left( \frac{1}{4}\alpha_{\mathrm{E, at}}\right)^2,
\end{gather*}
where $\alpha_{\mathrm{E, at}}=\nu_0^{-2}$.
\end{Thm}

\begin{rem}\quad
\begin{itemize}\itemsep=0pt
\item[1.] The constant $\alpha_{\mathrm{E, at}}$ is the dipole moment of a decoupled atom, i.e.,
\begin{gather}
\alpha_{\mathrm{E, at}}=
\frac{2}{3}\big\la \psi_{\mathrm{at}}|x\cdot (h_{\mathrm{at}}-3e \nu_0/2)^{-1}x\psi_{\mathrm{at}}\big\ra, \label{DMoment}
\end{gather}
where $h_{\mathrm{at}}=-\frac{1}{2}\Delta+\frac{e^2\nu_0^2}{2} x^2$ and $\psi_{\mathrm{at}}$ is the ground state of $h_{\mathrm{at}}$. Note that $x_j \psi_{\rm at}$ is orthogonal to $\psi_{\rm at}$: $\la \psi_{\rm at}|x_j\psi_{\rm at}\ra=0$. Thus, the vectors $(h_{\rm at}-3e\nu_0/2)^{-1}x_j\psi_{\rm at}$ in~(\ref{DMoment})
 are mathematically meaningful.
 \item[2.] The restrictions of the parameters in Theorem~\ref{Rto7} come from technical reasons: As we will see in the later sections, these are needed in order to control the perturbative expansions for~$E$ and~$E(R)$.
\end{itemize}
\end{rem}

\begin{exam}Let $\eta\in \mathscr{S}\big(\BbbR^3\big)$, the Schwartz space. Suppose that $\eta$ satisfies the following:
\begin{itemize}\itemsep=0pt
\item $\eta(0)=(2\pi)^{-3/2}$;
\item $\eta(k)$ is real-valued;
\item $\eta(k)=\eta_{\rm rad}(|k|)$.
\end{itemize}
For given $\xi>0$, we define $\varrho$ by
\begin{gather*}
\hat{\varrho}(k)=\eta(\xi k).
\end{gather*}
Then $\varrho$ satisfies (A.1)--(A.3). In addition, since
\begin{gather*}
\|\hat{\varrho}\|_{L^2}\propto \xi^{-3/2},\qquad \||k| \hat{\varrho}\|_{L^2} \propto \xi^{-5/2},\qquad \|
|k|^{-1}\hat{\varrho}
\|_{L^2} \propto \xi^{-1/2},
\end{gather*}
 the all assumptions in Theorem \ref{Rto7} are fulfilled, provided that $\xi$ is large enough. Note that a~typical choice of $\eta$ is $\eta(k)=(2\pi)^{-3/2}\ex^{-|k|^2}$.
\end{exam}

\section{Feynman Hamiltonians} \label{FeynRep}

\subsection{Preliminaries}
To prove our main result, let us introduce Feynman Hamiltonians of the nonrelativistic QED~\cite{Feynman}.
These Hamiltonians can be diagonalized readily as we will see in Sections~\ref{Dai1} and~\ref{Dai2}.

First, remark the following identification:
 \begin{gather*}
 L^2\big(\BbbR^3\big)=L^2_e\big(\BbbR^3\big) \oplus L^2_o\big(\BbbR^3\big),
 \end{gather*}
 where
 \begin{gather*}
 L_e^2\big(\BbbR^3\big)=\big\{f\in L^2\big(\BbbR^3\big)\, |\, f(-x)=f(x)\ \text{a.e.} \big\},\\
 L^2_o\big(\BbbR^3\big) =\big\{f\in L^2\big(\BbbR^3\big)\, |\, f(-x)=-f(x)\ \text{a.e.} \big\}.
 \end{gather*}
 For notational convenience, we denote by $\vepsilon_j(\cdot, \lambda)$ the multiplication operator by the function~$\vepsilon_j(\cdot, \lambda)$.

 We begin with the following lemma.
 \begin{lemm} Let
 \begin{alignat*}{3}
& \mathfrak{H}_1 =\bigcup_{j=1}^3 \overline{\R}\big(\vepsilon_j(\cdot, 1) \restriction L^2_e\big(\BbbR^3\big)\big),
 \qquad && \mathfrak{H}_2 =\bigcup_{j=1}^3 \overline{\R}\big(\vepsilon_j(\cdot, 2) \restriction L^2_e\big(\BbbR^3\big)\big),& \\
& \mathfrak{H}_3 =\bigcup_{j=1}^3 \overline{\R}\big(\vepsilon_j(\cdot, 1) \restriction L^2_o\big(\BbbR^3\big)\big),
\qquad && \mathfrak{H}_4 =\bigcup_{j=1}^3 \overline{\R}\big(\vepsilon_j(\cdot, 2) \restriction L^2_o\big(\BbbR^3\big)\big),&
 \end{alignat*}
 where $A\restriction \mathfrak{X}$ indicates the restriction of~$A$ to $\mathfrak{X}$.
 Then $\mathfrak{H}_1$, $\mathfrak{H}_2$, $\mathfrak{H}_3$ and $\mathfrak{H}_4$ are subspaces of~$L^2\big(\BbbR^3\big)$.
 \end{lemm}
 \begin{proof}
 Let $\mathbb{D}=\big\{k\in \BbbR^3\, |\, k_1\neq 0, k_2\neq 0, k_3\neq 0\big\}$.
 Trivially, $\vepsilon_j(k, \lambda)$ is well-defined on $\mathbb{D}$. In addition,
 $\vepsilon_1(k, 1)^{-1}$, $\vepsilon_2(k, 1)^{-1}$, $\vepsilon_1(k, 2)^{-1}$, $\vepsilon_2(k, 2)^{-1} $ and $\vepsilon_3(k, 2)^{-1}$ are well-defined on $\mathbb{D}$.\footnote{These facts immediately follow from
 (\ref{PolariDef}). Here, note that $\vepsilon(k, 2)$ is written as
 $\vepsilon(k, 2)=\big(k_1k_3, k_2k_3, -k^2_1-k_2^2\big)\big/|k|\sqrt{k_1^2+k_2^2}.$}
 Let~$C_0(\mathbb{D})$ be the set of continuous functions on $\mathbb{D}$ of compact support.
 Because the Lebesgue measure of $\mathbb{D}^c$, the complement of $\mathbb{D}$, is equal to zero, $C_0(\mathbb{D})$ is dense in $L^2\big(\BbbR^3\big)$. Thus, it holds that
 \begin{gather}
 \overline{\R}\big(
 \vepsilon_j(\cdot, 1) \restriction L^2_e\big(\BbbR^3\big)
 \big)
 =
 \overline{\R}\big(
 \vepsilon_j(\cdot, 1) \restriction C_{0, e}(\mathbb{D})
 \big),\qquad j=1, 2, \label{RanEquiv}
 \end{gather}
 where $C_{0, e}(\mathbb{D})=\{f\in C_0(\mathbb{D})\, |\, f(-k)=f(k)\}$.

 Let $F, G\in \mathfrak{H}_1$. Then there exist $i, j\in \{1, 2\}$ such that
 $F\in \overline{\R}\big( \vepsilon_j(\cdot, 1) \restriction L^2_e\big(\BbbR^3\big)
 \big)$ and $G\in \overline{\R}\big(
 \vepsilon_i(\cdot, 1) \restriction L^2_e\big(\BbbR^3\big)
 \big)$. By (\ref{RanEquiv}), there exist approximating sequences
 $(F_n)\subset \R\big(
 \vepsilon_j(\cdot, 1) \restriction C_{0, e}(\mathbb{D})
 \big)
 $ and $(G_n)\subset \R\big(
 \vepsilon_i(\cdot, 1) \restriction C_{0, e}(\mathbb{D})
 \big)
 $ such that $\|F-F_n\| \to 0$ and $\|G-G_n\| \to 0$ as $n\to 0$.
 Hence, for each $\alpha, \beta\in \BbbC$, it holds that
 \begin{gather}
 \alpha F_n+\beta G_n \to \alpha F+\beta G\qquad \mbox{as $n\to \infty$}. \label{LConv}
 \end{gather}
 Note that we can write $F_n=\vepsilon_j(\cdot, 1) f_n$ and $G_n=\vepsilon_i(\cdot, 1) g_n$
 with $f_n, g_n\in C_{0, e}(\mathbb{D})$. Thus, we have
 $G_n=\vepsilon_j(\cdot, 1) g_n'$, where $g_n'=\vepsilon_j(\cdot, 1)^{-1} \vepsilon_i(\cdot, 1) g_n$.
 Because $\vepsilon_j(k, 1)^{-1} \vepsilon_i(k, 1)$ is an even function on $\mathbb{D}$, we see that
 $g_n'\in C_{0, e}(\mathbb{D})$. Accordingly,
 $ \alpha F_n+\beta G_n=\vepsilon_j(\cdot, 1) (\alpha f_n+\beta g_n')
 \in \R\big( \vepsilon_j(\cdot, 1) \restriction C_{0, e}(\mathbb{D}) \big)
 $. Combining this, (\ref{RanEquiv}) and (\ref{LConv}), we conclude that
 $\alpha F+\beta G\in \mathfrak{H}_1$, in particular,
 $\mathfrak{H}_1$~is a subspace of~$L^2\big(\BbbR^3\big)$. By similar arguments, we can prove that $
 \mathfrak{H}_2, \mathfrak{H}_3$ and $\mathfrak{H}_4$ are subspaces of~$L^2\big(\BbbR^3\big)$. \end{proof}

\begin{lemm} We have the following identifications:
 \begin{gather}
 L^2\big(\BbbR^3\times \{1, 2\}\big) =L^2\big(\BbbR^3\big)\oplus L^2\big(\BbbR^3\big)
 =\mathfrak{H}_1\oplus \mathfrak{H}_2\oplus\mathfrak{H}_3\oplus \mathfrak{H}_4. \label{4DirectSum}
 \end{gather}
 \end{lemm}
 \begin{proof} The first identification in (\ref{4DirectSum}) is trivial.
 In what follows, we will concentrate on the proof of the second identification.

 Note that the multiplication operator $\vepsilon_j(\cdot, 1)^{-1}\, (j=1, 2)$ is self-adjoint and
 $\D\big(\vepsilon_j(\cdot, 1)^{-1}\big)$ is dense in $L^2\big(\BbbR^3\big)$.
 Because $\R(\vepsilon_j(\cdot, 1))\supseteq \D\big(
 \vepsilon_j(\cdot, 1)^{-1}\big) \supseteq C_0(\mathbb{D})$ for $j=1, 2$, we obtain $\overline{\R}(\vepsilon_j(\cdot, 1))=L^2\big(\BbbR^3\big)$.
 For each $f\in L^2\big(\BbbR^3\big)$, we set $f_e(k)=\frac{1}{2}(f(k)+f(-k))$ and $f_o(k)=\frac{1}{2}(f(k)-f(-k))$.
 Because $\vepsilon_j(k, 1)^{2}$ is an even function, we have
 $\la \vepsilon_j(\cdot, 1)f_e|\vepsilon_j(\cdot, 1)f_o\ra=0$ for all $f\in L^2\big(\BbbR^3\big)$ and $j=1, 2$, which implies that $\overline{\R}\big(\vepsilon_j(\cdot, 1)\restriction L^2_e\big(\BbbR^3\big)\big)\perp \overline{\R}\big(\vepsilon_j(\cdot, 1) \restriction L^2_o\big(\BbbR^3\big)\big)$. Since
 \begin{gather*}
 \underbrace{\vepsilon_j(\cdot, 1)f}_{\in \R(\vepsilon_j(\cdot, 1))}=
 \underbrace{\vepsilon_j(\cdot, 1)f_e}_{\in \R(\vepsilon_j(\cdot, 1)\restriction L_e^2(\BbbR^3) )}+
 \underbrace{\vepsilon_j(\cdot, 1)f_o}_{\in \R(\vepsilon_j(\cdot, 1)\restriction L_o^2(\BbbR^3) )},
 \qquad f\in L^2\big(\BbbR^3\big),
 \end{gather*}
 we conclude that
 \begin{gather*}
 L^2\big(\BbbR^3\big)=\overline{\R}\big(\vepsilon_j(\cdot, 1)\restriction L^2_e\big(\BbbR^3\big)\big)\oplus \overline{\R}\big(\vepsilon_j(\cdot, 1) \restriction L^2_o\big(\BbbR^3\big)\big)
 \end{gather*} for $j=1, 2$.

 For $i, i\rq{}\in \{1, 2\}$, we set $ \mu_{ii\rq{}}^{(1)}(k)=\vepsilon_i(k, 1) \vepsilon_{i\rq{}}(k, 1)$.
 Because $\mu_{ii'}^{(1)}(k)$ is an even function,
 we see that, for each $f\in L_e^2\big(\BbbR^3\big)$ and $g\in L_o^2\big(\BbbR^3\big)$,
 \begin{gather*}
 \la \vepsilon_i(\cdot, 1) f|\vepsilon_{i'}(\cdot, 1) g\ra=\big\la f|\mu_{ii'}^{(1)} g\big\ra=0.
 \end{gather*}
 Therefore, $\mathfrak{H}_1\perp \mathfrak{H}_3$ holds.
 Because $\mathfrak{H}_1\oplus \mathfrak{H}_3 \supseteq \overline{\R}\big(\vepsilon_j(\cdot, 1)\restriction L^2_e\big(\BbbR^3\big)\big)\oplus \overline{\R}\big(\vepsilon_j(\cdot, 1) \restriction L^2_o\big(\BbbR^3\big)\big)$, we finally arrive at $L^2\big(\BbbR^3\big)=\mathfrak{H}_1\oplus \mathfrak{H}_3$.
 By arguments similar to the above, we get that $
 L^2\big(\BbbR^3\big)=\mathfrak{H}_2\oplus \mathfrak{H}_4
 $. \end{proof}

We will construct a useful identification between $\Fock\big(L^2\big(\BbbR^3\big)\oplus L^2\big(\BbbR^3\big)\big)$ and $\bigotimes\limits_{\lambda=1}^4 \Fock(\mathfrak{H}_{\lambda})$ in Section~\ref{Sec34}.
 For this purpose, we recall some basic definitions in Sections \ref{Sec32} and \ref{Sec33}.

 \subsection[Second quantized operators in $\Fock\big(L^2\big(\BbbR^3\big) \oplus L^2\big(\BbbR^3\big)\big)$]{Second quantized operators in $\boldsymbol{\Fock\big(L^2\big(\BbbR^3\big) \oplus L^2\big(\BbbR^3\big)\big)}$}\label{Sec32}

 Let $a(f_1\oplus f_2)$ be the annihilation operator acting in $\Fock\big(L^2(\BbbR^3\times \{1, 2\}))=\Fock(L^2\big(\BbbR^3\big) \oplus L^2\big(\BbbR^3\big)\big)$. As usual, we express this operator as
 \begin{gather*}
 a(f_1\oplus f_2)=\sum_{\lambda=1, 2} \int_{\BbbR^3} \dm k\, f_{\lambda}^*(k) a(k, \lambda).
 \end{gather*}
 The Fock vacuum in $\Fock\big(L^2\big(\BbbR^3\big) \oplus L^2\big(\BbbR^3\big)\big)$ is denoted by $\Psi_0$. Let~$F$ be a real-valued function on~$\BbbR^3$ which is finite almost everywhere.
 The multiplication operator by $F$ is also written as~$F$. The second quantization of $F\oplus F$ is then given by
 \begin{gather*}
 \dG(F\oplus F)=0\oplus \Bigg[\bigoplus_{n=1}^{\infty} \sum_{j=1}^n 1\otimes \cdots \otimes \underbrace{(F\oplus F)}_{j^{\mathrm{th}}}\otimes \cdots \otimes 1\Bigg].
 \end{gather*}
 Needless to say, $\dG(F\oplus F)$ acts in $\Fock\big(L^2\big(\BbbR^3\big) \oplus L^2\big(\BbbR^3\big)\big)$. It is known that $\dG(F\oplus F)$ is essentially self-adjoint
 on a dense subspace
 \begin{gather*}
 \big\{\Psi=\{\Psi_n\}_{n=0}^{\infty}\, |\, \Psi_n\in (\D(F)\oplus \D(F))^{\odot n},
 \mbox{$\exists\, N\in \BbbN$ s.t.\ $\Psi_m=0$ $ \forall\, m>N$}
 \big\},
 \end{gather*}
 where $\odot $ indicates the algebraic tensor product. We will denote the closure of $\dG(F\oplus F)$ by the same symbol. Symbolically, we express $\dG(F\oplus F)$ as
 \begin{gather*}
 \dG(F\oplus F)=\sum_{\lambda=1, 2} \int_{\BbbR^3} \dm k\, F(k)a(k, \lambda)^*a(k, \lambda).
 \end{gather*}

\subsection[Second quantized operators in $\bigotimes\limits_{\lambda=1}^4\Fock(\mathfrak{H}_{\lambda})$]{Second quantized operators in $\boldsymbol{\bigotimes\limits_{\lambda=1}^4\Fock(\mathfrak{H}_{\lambda})}$}
 \label{Sec33}

 Let $a_{\lambda}(f_{\lambda})$ be the annihilation operator on $\Fock(\mathfrak{H}_{\lambda})$.
 We employ the following identifications:
 $ a_1(f_1)=a_1(f_1)\otimes \one \otimes \one \otimes \one$, $a_2(f_2)=\one \otimes a_2(f_2) \otimes \one \otimes \one $ and so on. Thus, $a_{\lambda}(f_{\lambda})$ can be regarded as a linear operator
 acting in the Hilbert space $\bigotimes\limits_{\lambda=1}^4\Fock(\mathfrak{H}_{\lambda})$.
 Let~$F$ be a real-valued function on~$\BbbR^3$. Suppose that $F$ is even: $F(-k)=F(k)$ a.e..
 $\dG_{\lambda}(F)$ denotes the second quantization of $F$ which acts in $\Fock(\mathfrak{H}_{\lambda})$.
 As before, we can also regard $\dG_{\lambda}(F)$ as a linear operator acting in $\bigotimes\limits_{\lambda=1}^4\Fock(\mathfrak{H}_{\lambda})$.
 The Fock vacuum in $\Fock(\mathfrak{H}_{\lambda})$ is denoted by $\Psi_{\lambda}$.
 We will freely use the following notations:
 \begin{gather*}
 a_{\lambda}(f) = \int_{\BbbR^3} \dm k\, f(k)^* a(k, \lambda),\qquad f\in \mathfrak{H}_{\lambda},\\
 \dG_{\lambda}(F) =\int_{\BbbR^3} \dm k\, F(k)a^*(k, \lambda)a(k, \lambda),\qquad \lambda=1, 2, 3, 4.
 \end{gather*}

 \subsection[Identifications between $\Fock\big(L^2\big(\BbbR^3\big)\oplus L^2\big(\BbbR^3\big)\big)$
 and $\bigotimes\limits_{\lambda=1}^4 \Fock(\mathfrak{H}_{\lambda})$]{Identifications between $\boldsymbol{\Fock\big(L^2\big(\BbbR^3\big)\oplus L^2\big(\BbbR^3\big)\big)}$
 and $\boldsymbol{\bigotimes\limits_{\lambda=1}^4 \Fock(\mathfrak{H}_{\lambda})}$}\label{Sec34}

 For each ${\bs f}=(f_1, f_2)\in L^2\big(\BbbR^3\big) \oplus L^2\big(\BbbR^3\big)$, we set
 \begin{gather}
 b_{ij}({\bs f})=a\big( \vepsilon_i(\cdot, 1) f_1 \oplus \vepsilon_j(\cdot, 2)f_2 \big),\nonumber\\
 c_{ij}({\bs f})=a_1(\vepsilon_i(\cdot, 1)f_{1, e})+a_2(\vepsilon_j(\cdot, 2)f_{2, e})\label{EquiA}
 -\im a_3(\vepsilon_i(\cdot, 1)f_{1, o})-\im a_4(\vepsilon_j(\cdot, 2) f_{2, o}),
 \end{gather}
 where $f_e(k)=(f(k)+f(-k))/2$ and $f_o(k)=(f(k)-f(-k))/2$.
 Let $\Psi_0$ be the Fock vacuum in $\Fock\big(L^2\big(\BbbR^3\big)\oplus L^2\big(\BbbR^3\big)\big)\colon \Psi_0=1\oplus 0 \oplus 0 \oplus \cdots$.
 \begin{lemm}\label{FockIden}
 We define a linear operator $V\colon \Fock(L^2\big(\BbbR^3\big) \oplus L^2\big(\BbbR^3\big)) \to \bigotimes\limits_{\lambda=1}^4 \Fock( \mathfrak{H}_{\lambda})$ by
 \begin{gather*}
 V\Psi_0=\bigotimes_{\lambda=1}^4\Psi_{\lambda},\\
 V\left[\prod_{\ell=1}^N b_{i_{\ell} j_{\ell}}({\bs f}_{\ell})^*\right] \Psi_0
 =\left[\prod_{\ell=1}^N c_{i_{\ell} j_{\ell}}({\bs f}_{\ell})^*\right] \bigotimes_{\lambda=1}^4\Psi_{\lambda}
 \end{gather*}
 for each ${\bs f}_1, \dots, {\bs f}_N\in L^2\big(\BbbR^3\big) \oplus L^2\big(\BbbR^3\big)$
 and $N\in \BbbN$. Then $V$ can be extended to the unitary operator. In what follows, we denote the extension by the same symbol. Then we have
 \begin{gather}
 Vb_{ij}({\bs f})V^{-1}=\overline{c_{ij}({\bs f})} \label{AnCrEq}
 \end{gather}
 for each ${\bs f} \in L^2\big(\BbbR^3\big) \oplus L^2\big(\BbbR^3\big)$ and $i, j\in \{1, 2, 3\}$,
 where the bar indicates the closure of the operator.
 \end{lemm}
 \begin{proof} For $i, i\rq{}\in \{1, 2, 3\}$, we set
 \begin{gather*}
 \mu_{ii\rq{}}^{(1)}(k)=\vepsilon_i(k, 1) \vepsilon_{i\rq{}}(k, 1),\qquad
 \mu_{ii\rq{}}^{(2)}(k)=\vepsilon_i(k, 2) \vepsilon_{i\rq{}}(k, 2).
 \end{gather*}
 For ${\bs f}, {\bs f}\rq{}\in L^2\big(\BbbR^3\big)\oplus L^2\big(\BbbR^3\big)$ and
 $i, j, i\rq{}, j\rq{}\in \{1, 2, 3\}$, define
 \begin{gather*}
 D({\bs f}|{\bs f}\rq{})_{ij; i\rq{} j\rq{}}=
 \big\la f_1|\mu_{ii\rq{}}^{(1)}f_1\rq{} \big\ra+ \big\la f_2|\mu_{ii\rq{}}^{(2)}f_2\rq{} \big\ra.
 \end{gather*}

 First, we prove that $\big\{b_{ij}({\bs f})| {\bs f}\in L^2\big(\BbbR^3\big)\oplus L^2\big(\BbbR^3\big),\ i, j\in \{1, 2, 3\} \big\}$ and $\big\{c_{ij}({\bs f})| {\bs f}\in L^2\big(\BbbR^3\big)\oplus L^2\big(\BbbR^3\big),\ i, j\in \{1, 2, 3\}\big\}$ satisfy the similar commutations relations, that is,
 \begin{gather*}
 [b_{ij}(\bs f), b_{i\rq{}j\rq{}}({\bs f}\rq{})^*]=D({\bs f}|{\bs f}\rq{})_{ij; i\rq{} j\rq{}},\qquad
 [b_{ij}(\bs f), b_{i\rq{}j\rq{}}({\bs f}\rq{})]=0 
 \end{gather*}
 and
 \begin{gather*}
 [c_{ij}(\bs f), c_{i\rq{}j\rq{}}({\bs f}\rq{})^*]=D({\bs f}|{\bs f}\rq{})_{ij; i\rq{} j\rq{}},\qquad
 [c_{ij}(\bs f), c_{i\rq{}j\rq{}}({\bs f}\rq{})]=0.
 \end{gather*}
 To see this, note that $\mu_{ii\rq{}}^{(1)}(k)$ and $\mu_{ii\rq{}}^{(2)}(k)$
 are even functions. Thus,
 \begin{gather*}
 \big\la f_e|\mu_{ii\rq{}}^{(1)} g_o\big\ra=0=\big\la f_e|\mu_{ii\rq{}}^{(2)} g_o\big\ra,\qquad f, g\in L^2\big(\BbbR^3\big).
 \end{gather*}
 Accordingly, we have
 \begin{align*}
 [c_{ij}(\bs f), c_{i\rq{}j\rq{}}({\bs f}\rq{})^*]&=\big\la f_{1, e}|\mu_{ii\rq{}}^{(1)} f_{1, e}\rq{}\big\ra
 +\big\la f_{2, e}|\mu_{ii\rq{}}^{(2)} f_{2, e}\rq{}\big\ra
 +\big\la f_{1, o}|\mu_{ii\rq{}}^{(1)} f_{1, o}\rq{}\big\ra
 +\big\la f_{2, 0}|\mu_{ii\rq{}}^{(2)} f_{2, o}\rq{}\big\ra\nonumber\\
 &=D({\bs f}|{\bs f}\rq{})_{ij;i\rq{}j\rq{}}.
 \end{align*}
 To check other commutation relations are easy.

 Using the above fact, we readily confirm that
 \begin{gather*}
 \Bigg\la \Bigg[\prod_{\ell=1}^N b_{i_{\ell} j_{\ell}}({\bs f}_{\ell})^*\Bigg] \Psi_0\Bigg|
 \Bigg[\prod_{\ell=1}^{N\rq{} } b_{i_{\ell}\rq{} j_{\ell}\rq{}}({\bs f}_{\ell}\rq{})^*\Bigg] \Psi_0\Bigg\ra\\
 \qquad{} = \Bigg\la \Bigg[\prod_{\ell=1}^N c_{i_{\ell} j_{\ell}}({\bs f}_{\ell})^*\Bigg] \bigotimes_{\lambda=1}^4\Psi_{\lambda}\Bigg|
 \Bigg[\prod_{\ell=1}^{N\rq{} } c_{i_{\ell}\rq{} j_{\ell}\rq{}}({\bs f}_{\ell}\rq{})^*\Bigg] \bigotimes_{\lambda=1}^4\Psi_{\lambda}\Bigg\ra
 \end{gather*}
 for every ${\bs f}_1, \dots {\bs f}_N, {\bs f}_1\rq{}, \dots, {\bs f}_{N\rq{}}\rq{}\in L^2\big(\BbbR^3\big)\oplus L^2\big(\BbbR^3\big)$ and $N, N\rq{}\in \BbbN$.
 From (\ref{4DirectSum}), it follows that
 the subspace spanned by the set of vectors $
 \Big\{
 \Big[\prod\limits_{\ell=1}^N b_{i_{\ell} j_{\ell}}({\bs f}_{\ell})^*\Big] \Psi_0
 \Big\}
 $ is dense in $\Fock\big(L^2\big(\BbbR^3\big) \oplus L^2\big(\BbbR^3\big)\big)$
 and the subspace spanned by
 the set of vectors $ \Big\{\Big[\prod\limits_{\ell=1}^N c_{i_{\ell} j_{\ell}}({\bs f}_{\ell})^*\Big] \bigotimes\limits_{\lambda=1}^4\Psi_{\lambda}\Big\}
 $ is dense in $\bigotimes\limits_{\lambda=1}^4 \mathfrak{H}_{\lambda}$.
 Hence, $V$ can be extended to the unitary operator. To check~(\ref{AnCrEq}) is easy.
 \end{proof}

\begin{lemm}\label{aoaoa} Let $F$ be a real-valued even function on $\BbbR^3$. Assume that $F$
 is continuous. Then we obtain
 \begin{gather*}
 V\dG(F \oplus F)V^{-1} =\sum_{\lambda=1}^4 \dG_{\lambda}(F).
 \end{gather*}
 \end{lemm}
 \begin{proof} For readers\rq{} convenience, we will provide a~sketch of the proof. We will continue to use the notations in the proof of Lemma~\ref{FockIden}. Set
 \begin{gather*}
 {\bs B}_{{\bs i}{\bs j}}({\bs f}_1, \dots, {\bs f}_N) =\Bigg[\prod_{\ell=1}^N b_{i_{\ell} j_{\ell}}({\bs f}_{\ell})^*\Bigg] \Psi_0,\\
 {\bs C}_{{\bs i}{\bs j}}({\bs f}_1, \dots, {\bs f}_N) =\Bigg[\prod_{\ell=1}^N c_{i_{\ell} j_{\ell}}({\bs f}_{\ell})^*\Bigg] \bigotimes_{\lambda=1}^4\Psi_{\lambda}
 \end{gather*}
 for ${\bs f}_1, \dots, {\bs f}_N\in L^2\big(\BbbR^3\big) \oplus L^2\big(\BbbR^3\big)$, $
 {\bs i}=(i_1, \dots, i_N)$, ${\bs j}=(j_1, \dots, j_N)\in \{1, 2, 3\}^N
 $ and $N\in \BbbN$.
 We define dense subspaces of $\Fock\big(L^2\big(\BbbR^3\big)\oplus L^2\big(\BbbR^3\big)\big)$
and $\bigotimes\limits_{\lambda=1}^4\mathfrak{H}_{\lambda}$ by
 \begin{gather}
 \mathfrak{V}_1=\operatorname{Lin}\big\{
 {\bs B}_{{\bs i}{\bs j}}({\bs f}_1, \dots, {\bs f}_N)\, \big|\, {\bs f}_1, \dots, {\bs f}_N\in C_0\big(\BbbR^3\big)\!\oplus\! C_0\big(\BbbR^3\big),\, {\bs i}, {\bs j} \!\in\! \{1, 2, 3\}^N,\, N\in \BbbN
 \big\},\nonumber\\
 \mathfrak{V}_2=\operatorname{Lin}\big\{
 {\bs C}_{{\bs i}{\bs j}}({\bs f}_1, \dots, {\bs f}_N)\, \big|\, {\bs f}_1, \dots, {\bs f}_N\in C_0\big(\BbbR^3\big)\!\oplus\! C_0\big(\BbbR^3\big),\, {\bs i}, {\bs j} \!\in\! \{1, 2, 3\}^N,\, N\in \BbbN
 \big\}, \!\!\!\!\label{DenseSubS}
 \end{gather}
 where $\operatorname{Lin}(S)$ indicates the linear span of~$S$. As is well-known, $\dG(F\oplus F)$ and $\sum\limits_{\lambda=1}^4\dG_{\lambda}(F)$
 are essentially self-adjoint on $\mathfrak{V}_1$ and $\mathfrak{V}_2$, respectively.
 We readily confirm that
 \begin{gather*}
 \dG(F\oplus F) {\bs B}_{{\bs i}{\bs j}}({\bs f}_1, \dots, {\bs f}_N)
 = \sum_{\alpha=1}^N {\bs B}_{{\bs i}{\bs j}}({\bs f}_1, \dots, F\oplus F{\bs f}_{\alpha}, \dots, {\bs f}_N),\\
 \sum_{\lambda=1}^4\dG_{\lambda}(F) {\bs C}_{{\bs i}{\bs j}}({\bs f}_1, \dots, {\bs f}_N)
 = \sum_{\alpha=1}^N {\bs C}_{{\bs i}{\bs j}}({\bs f}_1, \dots, F\oplus F{\bs f}_{\alpha}, \dots, {\bs f}_N).
 \end{gather*}
 Therefore, by Lemma~\ref{FockIden}, we obtain
 \begin{gather*}
 V\dG(F\oplus F) V^{-1}=\sum_{\lambda=1}^4 \dG_{\lambda}(F)\qquad \mbox{on $\mathfrak{V}_2$}.
 \end{gather*}
 This concludes the proof of Lemma~\ref{aoaoa}. \end{proof}

 \subsection{Definition of the Feynman Hamiltonians}\label{Sec35}
 In this subsection, we introduce the Feynman Hamiltonians. To this end, let
 \begin{gather*}
\A(x)= \int_{\BbbR^3}\dm k\, \frac{\hat{\varrho}(k)}{\sqrt{2|k|}}
\big\{
\vepsilon(k, 1)\big(
a(k, 1)\cos (k\cdot x)+a(k, 3)\sin (k\cdot x)
\big)\\
\hphantom{\A(x)=}{} +\vepsilon(k, 2)\big(
a(k, 2)\cos (k\cdot x)+a(k, 4)\sin (k\cdot x)
\big)+\mathrm{h.c.}\big\},\\
H_0= \sum_{\lambda=1}^4\int_{\BbbR^3}\dm k\, |k|a(k, \lambda)^* a(k, \lambda).
\end{gather*}
Here, h.c.\ denotes the hermite conjugates of the preceeding terms.
Note that $\mathscr{A}(x)$ is essentially self-adjoint on $\mathfrak{V}_2$ defined by~(\ref{DenseSubS}). We denote its closure by the same symbol.
By~(\ref{EquiA}) and~(\ref{AnCrEq}), we have the following:
 \begin{gather*}
 V A(x)V^{-1}=\A(x),\qquad V\Hf V^{-1}=H_0.
 \end{gather*}
 Now we define the two-electron Feynman Hamiltonian $H_{\mathrm{F2e}}$ by
\begin{gather*}
H_{\mathrm{F2e}}= \frac{1}{2}\big({-}\im
 \nabla_1-e\A(0)\big)^2+\frac{1}{2}e^2\nu_0^2 x_1^2
+\frac{1}{2}\big({-}\im \nabla_2-e\A(r)\big)^2+\frac{1}{2}e^2\nu_0^2 x_2^2
\\
\hphantom{H_{\mathrm{F2e}}=}{} +e^2 \int_{\BbbR^3}\dm k\,
\hat{\varrho}(k)^2 \ex^{\im k\cdot r}\big(x_1\cdot \hat{k}\big) \big(x_2\cdot \hat{k}\big)+H_0.
\end{gather*}
 Remark that $H_{\mathrm{F2e}}$ acts in $L^2\big(\BbbR_{x_1}^3\big) \otimes L^2\big(\BbbR_{x_2}^3\big)\otimes \Big(
 \bigotimes\limits_{\lambda=1}^4\Fock(\mathfrak{H}_{\lambda}) \Big)$ and is bounded from below, provided that $R$ is sufficiently large.

The following proposition plays an important role in the present paper.
 \begin{Prop} If $R$ is large enough,
 $VH_{\mathrm{D2e}}V^{-1}= H_{\mathrm{F2e}}$.
 \end{Prop}
\begin{proof} Apply Lemmas \ref{FockIden} and \ref{aoaoa}. \end{proof}

As for the one-electron Feynman Hamiltonian, we obtain the following.
\begin{Prop}Let
\begin{gather*}
H_{\mathrm{F1e}}=\frac{1}{2}\big({-}\im \nabla-e\A(0)\big)^2+\frac{1}{2}e^2\nu_0^2 x^2+H_0.
\end{gather*}
We have $VH_{\mathrm{D1e}} V^{-1}=H_{\mathrm{F1e}}$.
\end{Prop}

In Remark \ref{WhyF?}, we will explain why the Feynman Hamiltonians are useful.

\section{Canonical transformations} \label{CanoTr}
Let $U$ be a unitary operator on $L^2\big(\BbbR_{x_1}^3\big) \otimes L^2\big(\BbbR_{x_2}^3\big)\otimes \Big(
 \bigotimes\limits_{\lambda=1}^4\Fock(\mathfrak{H}_{\lambda}) \Big)$ defined by
\begin{gather*}
U=\exp \{\im e x_1\cdot \A(0)+\im e x_2\cdot \A(r)\}.
\end{gather*}
Then one readily confirms that
\begin{gather*}
U^* (-\im \nabla_1) U=-\im \nabla_1+e\A(0),\qquad U^*(-\im \nabla_2)U=-\im \nabla_2+e\A(r)
\end{gather*}
and
\begin{gather*}
U^* a(k, \lambda) U=
\begin{cases}
a(k, \lambda)+\im e\dfrac{\hat{\varrho}(k)
}{\sqrt{2|k|}} \vepsilon(k,
 \lambda)\cdot (x_1+x_2 \cos (k\cdot r))
&\mbox{for $\lambda=1,2$},\vspace{1mm}\\
a(k, \lambda)+\im e \dfrac{\hat{\varrho}(k)}{\sqrt{2|k|}} \vepsilon(k,
 \lambda-2)\cdot x_2 \sin (k\cdot r)
&\mbox{for $\lambda=3,4$}.
\end{cases}
\end{gather*}
Here, we used the following fact:
\begin{gather*}
\ex^T a(k, \lambda)\ex^{-T}=a(k, \lambda)+G(k, \lambda),
\end{gather*}
where $T=\sum\limits_{\lambda=1}^4 \{a_{\lambda}(G(\cdot, \lambda))-a(G(\cdot, \lambda))^*\}^{**}$, $G(\cdot, \lambda)\in \mathfrak{H}_{\lambda}$.
Hence, we arrive at\footnote{The reason why the last term in the right-hand side of (\ref{TransformedH}) appears is as follows. After performing the unitary transformation, we see that $U^*H_{\mathrm{F2e}} U$ contains the term concerning $\big(x_1\cdot \hat{k}\big)\big(x_2\cdot \hat{k}\big)$ and $(\vepsilon(k, \lambda)\cdot x_1)(\vepsilon(k, \lambda)\cdot x_2)$, which is given by
\begin{gather}
e^2\int_{\BbbR^3}\dm k\, \hat{\varrho}(k)^2 \cos(k\cdot r)
\bigg\{
\big(x_1\cdot \hat{k}\big)\big(x_2\cdot \hat{k}\big)+\sum_{\lambda=1, 2} (\vepsilon(k, \lambda)\cdot x_1)(\vepsilon(k, \lambda)\cdot x_2)
\bigg\}. \label{x_1x_2Term}
\end{gather}
Here, we used the fact that $\int \dm k\, \hat{\varrho}(k)^2\sin(k\cdot r) \big(x_1\cdot \hat{k}\big)\big(x_2\cdot \hat{k}\big)=0$. By applying the basic property $
\sum\limits_{\lambda=1, 2} |\vepsilon(k, \lambda)\ra \la \vepsilon(k, \lambda)|=\one_3-|\hat{k}\ra\la \hat{k}|$, we conclude that~(\ref{x_1x_2Term}) is equal to $e^2\int_{\BbbR^3}\dm k\, \hat{\varrho}(k)^2\cos( k\cdot r) (x_1\cdot x_2)$.}
\begin{gather}
U^* H_{\mathrm{F2e}} U= -\frac{1}{2}\Delta_1+\frac{1}{2}e^2\nu^2x_1^2-\frac{1}{2}\Delta_2+\frac{1}{2}e^2\nu^2x_2^2
+ex_1\cdot E(0)+ex_2\cdot E(r)+H_0\nonumber\\
\hphantom{U^* H_{\mathrm{F2e}} U=}{} +e^2\int_{\BbbR^3}\dm k\,
 \hat{\varrho}(k)^2\cos( k\cdot r) (x_1\cdot x_2), \label{TransformedH}
\end{gather}
where $\nu^2=2\nu_0^2$ and
\begin{gather*}
E(x)= \im \int_{\BbbR^3} \dm k\, \sqrt{\frac{|k|}{2}} \hat{\varrho}(k)
\big\{\vepsilon(k,1)\big(\cos(k\cdot x) a(k, 1)
+ \sin (k\cdot x) a(k, 3)\big)\\
\hphantom{E(x)=}{} +\vepsilon(k,2)\big(\cos (k\cdot x) a(k, 2)
+ \sin (k\cdot x) a(k, 4)\big)-\mathrm{h.c.}\big\}.
\end{gather*}

Let $\Nf$ be the number operator defined by $\Nf=\sum\limits_{\lambda=1}^4 \dG_{\lambda}(1)$.
Applying the ``Fourier transformation'' $\ex^{-\im \pi \Nf/2}$
in the Fock space,\footnote{Let $\pi(k, \lambda)=-\frac{\im}{\sqrt{2}} (a(k, \lambda)-a(k, \lambda)^*)$ and $\phi(k, \lambda) =\frac{1}{\sqrt{2}}(a(k, \lambda)+a(k, \lambda)^*)$. We can confirm that $
 [\pi(k, \lambda), \phi(k', \lambda')]=-\im \delta_{\lambda\lambda'}\delta(k-k')
 $. Recalling the fact $[-\im {\rm d}/{\rm d}x, x]=-\im $, $\pi(k, \lambda)$ and $\phi(k, \lambda)$
 can be regarded as a multiplication operator and a differential operator, respectively.
Now, we readily check that $\ex^{\im \pi\Nf/2} \pi(k, \lambda)\ex^{-\im \pi \Nf/2}=\phi(k, \lambda)$ holds,
 which corresponds to the relation $\mathcal{F} x\mathcal{F}^{-1}=-\im {\rm d}/{\rm d}x$, where $\mathcal{F}$ is the Fourier transformation on~$L^2(\BbbR)$. This similarity is a reason why we refer to the unitary operator $\ex^{\im \pi\Nf/2}$ as the Fourier transformation.}
 we obtain that
\begin{gather}
H= \ex^{\im\pi \Nf/2}U^* H_{\mathrm{F2e}}U \ex^{-\im \pi \Nf/2}\nonumber\\
\hphantom{H}{} = -\frac{1}{2}\Delta_1+\frac{1}{2}e^2\nu^2x_1^2-\frac{1}{2}\Delta_2+\frac{1}{2}e^2\nu^2x_2^2+ex_1\cdot\hat{ E}(0)+ex_2\cdot \hat{E}(r)+H_0\nonumber\\
\hphantom{H=}{} +e^2\int_{\BbbR^3}\dm k\,
 \hat{\varrho}(k)^2\cos( k\cdot r) (x_1\cdot x_2),\label{DefH}
\end{gather}
where
\begin{gather}
\hat{E}(x)= \int_{\BbbR^3} \dm k\, |k| \hat{\varrho}(k)
\big\{\vepsilon(k,1)\big(\cos(k\cdot x) q(k, 1)
+ \sin (k\cdot x) q(k, 3)\big)\nonumber\\
\hphantom{\hat{E}(x)=}{} +\vepsilon(k,2)\big(\cos (k\cdot x) q(k, 2)
+ \sin (k\cdot x) q(k, 4)\big)\big\} \label{DefE}
\end{gather}
and
\begin{gather*}
q(k, \lambda)=\frac{1}{\sqrt{2|k|}}\big(a(k, \lambda)+a(k, \lambda)^*\big).
\end{gather*}
Since due to the assumption (A.3) the last term in~(\ref{DefH}) gives a rapidly decreasing contribution as a function of~$R$ to the ground state energy, we ignore this term from now on.

Finally, we define
\begin{gather*}
K=-\frac{1}{2} \Delta+\frac{1}{2}e^2\nu^2 x^2+ex\cdot \hat{E}(0)+H_0.
\end{gather*}
By an argument similar to the construction of $U$, we can construct a unitary operator $u$ on $L^2\big(\BbbR^3\big) \otimes \Big(\bigotimes\limits_{\lambda=1}^4\Fock(\mathfrak{H}_{\lambda})\Big)$ such that $K=\ex^{\im \pi \Nf/2} u H_{\mathrm{F1e}} u^{-1} \ex^{-\im \pi \Nf/2}$.

\section{Lattice approximated Hamiltonians}\label{FiniteVApp}

In order to exactly compute the ground state energies of $H$ and $K$, we will first introduce the lattice approximation of Hamiltonians. As we will see in later sections, the approximated Hamiltonians can be regarded as Hamiltonians of finite dimensional harmonic oscillator, which are exactly solvable.

For each $\Lambda>0$, let $\chi_{\Lambda}$ be an ultraviolet cutoff function given by $\chi_{\Lambda}(k)=1$ if $|k|\le \Lambda$, $\chi_{\Lambda}(k)=0$ otherwise. We define a linear operator $\hat{E}_{\Lambda}(x)$ by replacing $\hat{\varrho}(k)$ with $\hat{\varrho}(k) \chi_{\Lambda}(k)$ in the definition of $\hat{E}(x)$, i.e., the equation~(\ref{DefE}). We also define $H_{0, \Lambda}$ by
\begin{gather*}
H_{0, \Lambda}=\sum_{\lambda=1}^4 \int_{\BbbR^3} \dm k \, |k| \chi_{\Lambda}(k)a(k, \lambda)^*a(k, \lambda)
.
\end{gather*}
The Hamiltonians with a cutoff $\Lambda$ are defined by
\begin{gather*}
H_{\Lambda} =-\frac{1}{2}\Delta_{1}+\frac{1}{2}e^2\nu^2x_1^2 -\frac{1}{2}\Delta_{2}+\frac{1}{2}e^2\nu^2x_2^2+ex_1\cdot \hat{E}_{\Lambda}(0)+ex_2\cdot \hat{E}_{\Lambda}(r) +H_{0, \Lambda},\\
K_{\Lambda} =-\frac{1}{2}\Delta+\frac{1}{2}e^2\nu^2x^2+ex\cdot \hat{E}_{\Lambda}(0)+H_{0, \Lambda}.
\end{gather*}
We readily see that $H_{\Lambda}$ and $K_{\Lambda}$ respectively converge to~$H$ and~$K$ in the norm resolvent sense as $\Lambda\to \infty$.

Let $M$ be the (momentum) lattice with a cutoff $\Lambda$, namely,
\begin{gather*}
M=\big\{
l\in (2\pi \BbbZ/L)^3\, \big|\, |l_i|\le 2\pi \Lambda,\, i=1,2,3
\big\}\backslash \{0\}.
\end{gather*}
For later use, we label the elements of $M$ as
\begin{gather*}
M=\{ k_1, \dots, k_N\}.
\end{gather*}
Then the lattice approximated Hamiltonians are defined by
\begin{gather*}
H_{L, \Lambda}
=-\frac{1}{2}\Delta_1+\frac{1}{2}e^2\nu^2x_1^2-\frac{1}{2}\Delta_2+\frac{1}{2}e^2\nu^2x_2^2+ex_1\cdot\hat{ E}_{L, \Lambda}(0)+ex_2\cdot \hat{E}_{L,
 \Lambda}(r)+H_{0,L, \Lambda},\\
 K_{L, \Lambda} = -\frac{1}{2}\Delta+\frac{1}{2}e^2\nu^2 x^2+ex\cdot \hat{E}_{L, \Lambda}(0)+H_{0, L, \Lambda},
\end{gather*}
where
\begin{gather}
\hat{E}_{L, \Lambda}(x)= \left(\frac{2\pi}{L}\right)^{3/2}\sum_{k\in M} |k| \hat{\varrho}(k)
\big\{
\vepsilon(k,1)\big(\cos(k\cdot x) q(k,1)
+ \sin (k\cdot x) q(k, 3)\big)\nonumber\\
\hphantom{\hat{E}_{L, \Lambda}(x)=}{} +\vepsilon(k,2)\big(\cos (k\cdot x) q(k, 2)
+ \sin (k\cdot x) q(k, 4)\big)\big\},\label{EleFielD}\\
H_{0, L, \Lambda}=\frac{1}{2}\sum_{\lambda=1}^4\sum_{k\in M}
\big(p(k, \lambda)^2+|k|^2 q(k, \lambda)^2\big)-2\sum_{k\in M}|k|\nonumber
\end{gather}
with $p(k, \lambda)=\frac{1}{i} \sqrt{\frac{|k|}{2}}(a(k, \lambda)-a(k, \lambda)^*)$.
The lattice approximated operators act in the Hilbert space
$L^2\big(\BbbR_{x_1}^3\big)\otimes L^2\big(\BbbR_{x_2}^3\big) \otimes \Big(\bigotimes\limits_{\lambda=1}^4 \Fock(\mathfrak{H}_{L, \Lambda, \lambda})\Big)$ or $L^2\big(\BbbR^3\big) \otimes
\Big(\bigotimes\limits_{\lambda=1}^4 \Fock(\mathfrak{H}_{L, \Lambda, \lambda})\Big)$, where
$\mathfrak{H}_{L, \Lambda, \lambda}=\ell_*^2(M) \allowbreak \cap \mathfrak{H}_{\lambda}$.
Here, $\ell_*^2(M)$ is the $\ell^2(M)$ equipped with a modified norm
\begin{gather*}
\|f\|_*=\left(\frac{2\pi}{L}\right)^{3/4}\left(\sum_{k\in M} |f(k)|^2\right)^{1/2},
\end{gather*} and we regard $\ell^2_*(M)$ as a closed subspace of $L^2\big(\BbbR^3\big)$.
 Note that $p(k, \lambda)$ and $q(k, \lambda)$ are essentially self-adjoint on the finite particle subspace
 of $\bigotimes\limits_{\lambda=1}^4 \Fock(\mathfrak{H}_{L, \Lambda, \lambda}) $.
In what follows, we denote their closures by same symbols, respectively.
$q(k, \lambda)$ and $p(k, \lambda)$ is a canonical pair of the
photonic displacement coordinate and its conjugate momentum
 satisfying the standard commutation
relations:
\begin{gather*}
[p(k, \lambda), q(k', \lambda')] =-\im
 \delta_{kk'}\delta_{\lambda\lambda'},\\
[p(k, \lambda), p(k', \lambda')] =0=[q(k, \lambda), q(k', \lambda')].
\end{gather*}
Recall the identification $\Fock(\BbbC)=L^2(\BbbR)$.
Using this, we can naturally embed
 $\bigotimes\limits_{\lambda=1}^4\Fock(\mathfrak{H}_{L, \Lambda, \lambda})$ into $
\Big(\bigotimes\limits_{\lambda=1}^4L^2\big(\BbbR^3\big)\Big)^{\otimes \# M}$. In addition, $p(k, \lambda)$ and
$q(k, \lambda)$ can be regarded as the differential and multiplication
 operators, respectively.

The following proposition is a basis for our computation.
\begin{Prop}\label{ResolCon}For each $z\in \BbbC\backslash \BbbR$, one has
\begin{gather*}
\lim_{\Lambda\to \infty} \lim_{L\to \infty} (H_{L, \Lambda}-z)^{-1} =(H-z)^{-1},\\
\lim_{\Lambda\to \infty} \lim_{L\to \infty} (K_{L, \Lambda}-z)^{-1} =(K-z)^{-1}
\end{gather*}
 in the operator norm topology.
\end{Prop}
\begin{proof}See, e.g., \cite{AH, GJ}.
\end{proof}

\section{Diagonalization I: One-electron Hamiltonian}\label{Dai1}

In this section, we diagonalize the one-electron Hamiltonian $K_{L, \Lambda}$. To this end, let
\begin{alignat*}{3}
& F_x(k, 1)=\left(\frac{2\pi}{L}\right)^{3/2} |k|\hat{\varrho}(k)\cos
 (k\cdot x),\qquad && F_x(k, 2)= \left(\frac{2\pi}{L}\right)^{3/2}|k|\hat{\varrho}(k)\cos
 (k\cdot x),& \\
& F_x(k, 3)= \left(\frac{2\pi}{L}\right)^{3/2} |k|\hat{\varrho}(k)\sin
 (k\cdot x),\qquad && F_x(k, 4)= \left(\frac{2\pi}{L}\right)^{3/2}|k|\hat{\varrho}(k)\sin
 (k\cdot x).&
\end{alignat*}
We define a linear operator $ \mathbb{T}(x)$ from $\ell^2(M\times \{1, \dots, 4\})$ to $\BbbC^3$ by
\begin{gather*}
\mathbb{T}(x){\boldsymbol f} =\sum_{\lambda=1}^4
\sum_{k\in M} |\vepsilon(k, \lambda)\ra F_x(k, \lambda)f(k, \lambda)
\end{gather*}
for each ${\boldsymbol f}=\{f(k, \lambda)\,|\, k\in M, \lambda\in \{1, \dots,4\}\}\in \ell^2(M\times \{1, \dots, 4\})$.
Here, we used the following notation:
$\vepsilon(k, 3):=\vepsilon(k, 1)$ and $\vepsilon(k, 4):=\vepsilon(k, 2)$.
The adjoint of $\mathbb{T}(x)$ is denoted by~$\mathbb{T}^*(x)$. Note that
\begin{gather*}
(\mathbb{T}^*(x)a)(k, \lambda)=\la \vepsilon(k, \lambda)\,|\,a\ra_{3} F_x(k, \lambda),\qquad a\in \BbbC^3,
\end{gather*}
where $\la \cdot| \cdot \ra_3$ stands for the inner product in $\mathbb{C}^3$.

Using the above notations, the interaction term $x\cdot \hat{E}_{L, \Lambda}(r)$
in $K_{L, \Lambda}$ is expressed as $x\cdot \hat{E}_{L, \Lambda}(r)=
\la \mathbb{T}(r) {\boldsymbol q}\,|\, x\ra_3=\la x\,|\,\mathbb{T}(r){\boldsymbol q}\ra_3$, where
 ${\boldsymbol q}=\{q(k, \lambda)\, |\, k\in M, \lambda\in \{1, \dots, 4\}\}$.
On the other hand, the field energy can be represented by
\begin{gather*}
H_{0, L, \Lambda}=\frac{1}{2}\big({\boldsymbol p}^2+\la {\boldsymbol q}| S_0{\boldsymbol q}\ra\big)-\frac{1}{2}\operatorname{tr} \big[\sqrt{S_0}\big],
\end{gather*}
where
$ {\boldsymbol p}=\{p(k, \lambda)\, |\, k\in M, \lambda\in \{1, \dots, 4\}\}$
 and
\begin{gather*}
S_0=
\begin{pmatrix}
|k_1|^2 \one_4& & & & O\\
 &|k_2|^2\one_4& & & & \\
& & & \ddots & & \\
&O & & & |k_N|^2\one_4
\end{pmatrix}.
\end{gather*}
Hence, $K_{L, \Lambda}$ can be rewritten as
\begin{gather}
K_{L, \Lambda}=-\frac{1}{2}\Delta+\frac{1}{2}e^2\nu^2 x^2+e \la x\,|\,\mathbb{T}(r) {\boldsymbol q}\ra_3
+\frac{1}{2}\big(
{\boldsymbol p}^2+\la {\boldsymbol q}\,|\, S_0{\boldsymbol q}\ra\big)-\frac{1}{2}\operatorname{tr}
\big[\sqrt{S_0}\big]. \label{OnePartiK}
\end{gather}
By setting $\phi=(x, {\boldsymbol q})$ and $\pi=(-\im \nabla,
{\boldsymbol p})$, one sees that
\begin{gather*}
K_{L, \Lambda}=\frac{1}{2}\big(
\la \pi| \pi\ra+\la \phi\,|\,\omega\phi\ra
\big)
-\frac{1}{2}\operatorname{tr}\big[\sqrt{\omega_0}\big]+\frac{3}{2}e\nu,
\end{gather*}
where
\begin{gather*}
\omega= \omega_0+Q,\qquad
\omega_0= \begin{pmatrix}
e^2\nu^2 & 0\\
0& S_0
\end{pmatrix},\qquad
Q=e\begin{pmatrix}
0 & \mathbb{T}(r)\\
\mathbb{T}^*(r) & 0
\end{pmatrix}.
\end{gather*}
The following lemma is a basic input.
\begin{lemm}\label{LLL}
If $1\le \sqrt{2}e\nu_0$ and $ \sqrt{2}e\|\hat{\varrho}\|_*<1$, then $\omega \ge 0$.
\end{lemm}
\begin{proof} By (\ref{TT*}), we have $\|\mathbb{T}(r) {\boldsymbol f}\| \le \sqrt{2} \|\hat{\varrho}\|_*\big\|S_0^{1/2} {\boldsymbol f}\big\|$ for all ${\boldsymbol f}\in \ell^2(M\times \{1, \dots, 4\})$.
Hence, for all $\vphi=(a, {\boldsymbol f})\in \BbbC^3\oplus \ell^2(M\times \{1, \dots, 4\})$, we have,
by the Schwarz inequality,
\begin{align*}
|\la \vphi|Q\vphi\ra| &\le 2\sqrt{2}e \| \hat{\varrho}\|_* \|a\|_3 \big\|S_0^{1/2} {\boldsymbol f}\big\|\\
&\le \sqrt{2} e\| \hat{\varrho}\|_*\big(\|a\|^2_3+\big\|S_0^{1/2} {\boldsymbol f}\big\|^2\big)\\
&\le \sqrt{2} e\|\hat{\varrho}\|_*\la \vphi|\omega_0\vphi\ra,
\end{align*}
provided that $1\le e^2\nu^2$. This concludes the proof of Lemma \ref{LLL}. \end{proof}

Therefore, the ground state energy of $K_{L, \Lambda}$ is given by
the following formula.
\begin{Prop}\label{EnergyTr}
Let $E_{L, \Lambda}=\inf \operatorname{spec}(K_{L, \Lambda})$.
If $1\le \sqrt{2} e \nu_0$ and $\sqrt{2}e \|\hat{\varrho}\|_* <1$, then one has
\begin{gather*}
E_{L, \Lambda}=\frac{1}{2}\operatorname{tr}\big[\sqrt{\omega}-\sqrt{\omega_0}\big] +\frac{3}{2}e\nu.
\end{gather*}
\end{Prop}
\begin{proof}
We provide a sketch of the proof.
First, we diagonalize $\omega$ as
\[ \omega=U^{-1} \operatorname{diag}(\lambda_1, \dots, \lambda_{4N+3}) U,
\] where $U$ is a unitary matrix and $\lambda_1, \dots, \lambda_{4N+3}$ are positive
eigenvalues of $\omega$. By setting $\tilde{\phi}=U\phi$ and $\tilde{\pi}=U \pi$, we can express $K_{L, \Lambda}$ as
\begin{gather}
K_{L, \Lambda} =\frac{1}{2}\big\la \tilde{\pi}|\tilde{\pi}\big\ra
+\frac{1}{2}\big\la \tilde{\phi}|\operatorname{diag}(\lambda_1, \dots, \lambda_{4N+3})\tilde{\phi}\big\ra
-\frac{1}{2}\operatorname{tr}[\sqrt{\omega_0}]+\frac{3}{2}e\nu. \label{ReExK}
\end{gather}
Because $\tilde{\pi}_j$ and $ \tilde{\phi}_j$ satisfy the Weyl relation:
 $e^{\im t \tilde{\pi}_i} e^{\im s \tilde{\phi}_j}=e^{\im st \delta_{ij}} e^{\im s \tilde{\phi}_j} e^{\im t \tilde{\pi}_i}$, the von Neumann's uniqueness theorem \cite[Theorem~VIII.14]{ReSi1} tells us that there is a unitary operator $\tau\colon L^2\big(\BbbR^{4N+3}\big) \allowbreak \to L^2\big(\BbbR^{4N+3}\big)$ such that $\tau\tilde{\phi}_j\tau^{-1} =x_j$ and $\tau \tilde{\pi}_j\tau^{-1}=-\im \partial/\partial x_j$.
 Therefore, the right-hand side of~(\ref{ReExK}) can be regarded as a Hamiltonian for $4N+3$-dimensional harmonic oscillator. Since the lowest eigenvalue of the Hamiltonian $-\frac{1}{2}\Delta_j^2+\frac{\lambda_j}{2} x_j^2$ is equal to $\sqrt{\lambda_j}/2$, we obtain that
\begin{gather*}
E_{L, \Lambda} =\frac{1}{2}\sum_{j=1}^{4N+3} \sqrt{\lambda_j}-\frac{1}{2}\operatorname{tr}\big[\sqrt{\omega_0}\big]+\frac{3}{2}e\nu =\frac{1}{2}\operatorname{tr}\big[\sqrt{\omega}-\sqrt{\omega_0}\big]
+\frac{3}{2}e\nu.
\end{gather*}
This finishes the proof of Proposition \ref{EnergyTr}. \end{proof}

Applying the elementary fact
\begin{gather}
\frac{1}{\pi}\int_{-\infty}^{\infty}\dm s \, \frac{a}{s^2+a}=\sqrt{a}, \label{I11}
\end{gather}
 we have that
\begin{gather}
E_{L, \Lambda}= \frac{1}{2\pi}\int_{-\infty}^{\infty}\dm s\,
 \operatorname{tr}\big[
\omega \big(s^2+\omega\big)^{-1}-\omega_0 \big(s^2+\omega_0\big)^{-1}
\big]+\frac{3}{2}e\nu\nonumber\\
\hphantom{E_{L, \Lambda}}{} = \frac{1}{2\pi}\sum_{j=1}^{\infty}(-1)^j\int_{-\infty}^{\infty}\dm s\,
 \operatorname{tr}\big[
\omega_0 \big(s^2+\omega_0\big)^{-1}\big\{
Q \big(s^2+\omega_0\big)^{-1} \big\}^j\big]\nonumber\\
\hphantom{E_{L, \Lambda}=}{} +\frac{1}{2\pi}\sum_{j=1}^{\infty}(-1)^{j+1}\int_{-\infty}^{\infty}\dm s\,
 \operatorname{tr}\big[\big\{ Q \big(s^2+\omega_0\big)^{-1}\big\}^j\big]+\frac{3}{2}e\nu\no
\hphantom{E_{L, \Lambda}}{} = \frac{1}{2\pi}\sum_{j=1}^{\infty}(-1)^{j+1}\int_{-\infty}^{\infty}\dm
 s\, s^2
 \operatorname{tr}\big[\big(s^2+\omega_0\big)^{-1}\big\{Q \big(s^2+\omega_0\big)^{-1} \big\}^j\big]+\frac{3}{2}e\nu.
\label{Expansion}
\end{gather}
Since $Q$ is off-diagonal, (\ref{Expansion}) becomes
\begin{gather}
E_{L, \Lambda}= -\frac{1}{2\pi}\sum_{j=1}^{\infty}\int_{-\infty}^{\infty}\dm s
\, s^2 \operatorname{tr}\big[\big(s^2+\omega_0\big)^{-1}
Q (s)^{2j}\big]+\frac{3}{2}e\nu, \label{ESeries}
\end{gather}
where $Q(s)=\big(s^2+\omega_0\big)^{-1/2} Q\big(s^2+\omega_0\big)^{-1/2}$.
In what follows, we will examine the convergence of the right-hand side of~(\ref{ESeries}).
As we will see, this series absolutely converges and~(\ref{ESeries})
is rigorously justified if $\nu_0$ is large enough.

We begin with the following basic lemma.
\begin{lemm}\label{EstT}
We have the following
\begin{gather*}
\big\|
\mathbb{T}(x) \big(s^2+S_0\big)^{-1/2}
\big\| \le \sqrt{2} \|\hat{\varrho}\|_{*},
\qquad
\big\|\big(s^2+S_0\big)^{-1/2} \mathbb{T}^*(x)\big\| \le \sqrt{2} \|\hat{\varrho}\|_{*}
\end{gather*}
for all $x\in \BbbR^3$, where $\|f\|_{*}=\sqrt{\left(\frac{2\pi}{L}\right)^{3}\sum\limits_{k\in M} |f(k)|^2}$ for each $f\in \ell^2(M)$.
\end{lemm}
\begin{proof}
For each ${\boldsymbol f}\in \ell^2(M\times \{1, \dots, 4\})$, we have, by~(\ref{TT*}),
\begin{align*}
\big\|
\mathbb{T}(x) \big(s^2+S_0\big)^{-1/2} {\boldsymbol f}
\big\|^2&=\la {\boldsymbol f}\,|\,\mathbb{T}^*_s(x) \mathbb{T}_s(x) {\boldsymbol f}\ra\\
&=\big\la {\boldsymbol f}\,|\,\big(s^2+S_0\big)^{-1/2} \mathbb{M}(x, x)\big(s^2+S_0\big)^{-1/2} {\boldsymbol f}\big\ra\\
&\le \Bigg|
\sum_{k, \lambda} \frac{F_x(k, \lambda)}{\big(s^2+k^2\big)^{1/2}} f(k, \lambda)
\Bigg|^2\\
&\le \big\|
\big(s^2+k^2\big)^{-1/2} F_x \big\|^2\|{\boldsymbol f}\|^2.
\end{align*}
Because $\big\|\big(s^2+k^2\big)^{-1/2} F_x\big\|^2\le 2\|\hat{\varrho}\|^2_*$, we conclude that
$\big\|\mathbb{T}(x) \big(s^2+S_0\big)^{-1/2}
\big\| \le \sqrt{2} \|\hat{\varrho}\|_{*}$.
\end{proof}

\begin{lemm}\label{EstInt}Let
\begin{gather}
D(s)=2e^2s^2\big(s^2+e^2\nu^2\big)^{-1}\nonumber\\
\hphantom{D(s)=}{} \times \big\{
\big(s^2+e^2\nu^2\big)^{-1} \big\|
\big(s^2+|k|^2\big)^{-1/2} |k| \hat{\varrho}
\big\|_*^2+\big\|
\big(s^2+|k|^2\big)^{-1} |k| \hat{\varrho}
\big\|_*^2
\big\}. \label{DefD(s)}
\end{gather}
Then we have the following:
\begin{itemize}\itemsep=0pt
\item[{\rm (i)}] For all $s\in \BbbR$,
\begin{gather*}
s^2\operatorname{tr}\big[\big(s^2+\omega_0\big)^{-1} Q(s)^{2n}\big]
\le \left(
\frac{\sqrt{2}}{\nu} \|\hat{\varrho}\|_* \right)^{2n-2} D(s).
\end{gather*}
\item[{\rm (ii)}]
Let $a=\big(
\frac{\sqrt{2}}{\nu}\|\hat{\varrho}\|_*\big)^2$.
If $a<1$, then we have
\begin{gather*}
\sum_{n=1}^{\infty} s^2\operatorname{tr}\big[\big(s^2+\omega_0\big)^{-1}Q(s)^{2n} \big]
\le \frac{1}{1-a} D(s).
\end{gather*}
Remark that $\lim\limits_{L\to \infty} a\le
\big(\frac{\sqrt{2}}{\nu} \|\hat{\varrho}\|_{L^2}\big)^2\le c_{\infty}^2<1/4
$ holds for all $\Lambda>0$ by the assumption in Theorem~{\rm \ref{Rto7}}. Thus, the condition $a<1$ is satisfied provided that~$L$ is sufficiently large.
\item[{\rm (iii)}] $D(s)\in L^1(\BbbR)$ and
\begin{gather}
\frac{1}{2\pi} \int_{\BbbR} \dm s\, D(s) \le \frac{e}{\nu}\|\hat{\varrho}\|_*^2.\label{EstD}
\end{gather}
\end{itemize}
\end{lemm}
\begin{proof} We set $\mathbb{T}_s(r)=\mathbb{T}(r)\big(s^2+S_0\big)^{-1/2}$. First, consider the case where $n=1$. Because
\begin{gather}
Q(s)^2=e^2\big(s^2+e^2\nu^2\big)^{-1}
\begin{pmatrix}
\mathbb{T}_s(r) \mathbb{T}^*_s(r) & 0\\
0 & \mathbb{T}^*_s(r) \mathbb{T}_s(r)
\end{pmatrix}, \label{Q^2Formula}
\end{gather}
we obtain that
\begin{gather*}
 s^2
 \operatorname{tr}\big[\big(s^2+\omega_0\big)^{-1}
Q (s)^{2}\big]\\
\qquad {} = e^2s^2 \big(s^2+e^2\nu^2\big)^{-2}
 \operatorname{tr}\big[
 \mathbb{T}_s(r) \mathbb{T}_s^*(r)
 \big] + e^2s^2 \big(s^2+e^2\nu^2\big)^{-1}
 \operatorname{tr}\big[
 \big(s^2+S_0\big)^{-1}
 \mathbb{T}_s^*(r) \mathbb{T}_s(r)
 \big].
\end{gather*}
By (\ref{TT*}) and (\ref{T*T}), we have
\begin{gather*}
\operatorname{tr}\big[
 \mathbb{T}_s(r) \mathbb{T}_s^*(r)
 \big]
 \le 2 \left(\frac{2\pi}{L}\right)^{3}\sum_{k\in M} \big(s^2+|k|^2\big)^{-1}|k|^2 |\hat{\varrho}(k)|^2,\\
 \operatorname{tr}\big[
 \big(s^2+S_0\big)^{-1}
 \mathbb{T}_s^*(r) \mathbb{T}_s(r)
 \big]
 \le 2 \left(\frac{2\pi}{L}\right)^{3}\sum_{k\in M}\big(s^2+|k|^2\big)^{-2}|k|^2|\hat{\varrho}(k)|^2.
\end{gather*}
Thus, we get (i) for $n=1$.

To prove the assertion for $n\ge 2$, we remark that $\|Q(s)\| \le \frac{\sqrt{2}}{\nu} \|\hat{\varrho}\|_{*}$, which immediately follows from Lemma~\ref{EstT} and~(\ref{Q^2Formula}). Thus, by using the fact $Q(s)^{2n} \le \|Q(s)\|^{2n-2} Q(s)^2$, we have
\begin{gather*}
\operatorname{tr}\big[\big(s^2+\omega_0\big)^{-1/2} Q(s)^{2n}\big(s^2+\omega_0\big)^{-1/2}\big]\!
\le \|Q(s)\|^{2n-2}\!\operatorname{tr}\big[\big(s^2+\omega_0\big)^{-1/2} Q(s)^2\big(s^2+\omega_0\big)^{-1/2}
\big].
\end{gather*}
Applying the result for $n=1$, we get the desired result for $n\ge 2$.
(ii) immediately follows from~(i).

By using the formula (\ref{I111}) with $a=b=e^2\nu^2$ and $c=|k|^2$, we see that
\begin{gather*}
 \frac{1}{2\pi}\int_{-\infty}^{\infty}\dm s \, 2 e^2 s^2 \big(s^2+e^2\nu^2\big)^{-2}
 \big\| \big(s^2+|k|^2\big)^{-1/2} |k|\hat{\varrho} \big\|_*^2 \\
\qquad{} =
 \frac{e^2}{2\pi}\left(\frac{2\pi}{L}\right)^{3}\sum_{k\in M}\frac{2 \pi}{2e\nu}\frac{|k|^2|\hat{\varrho}(k)|^2}{(e\nu+|k|)^2}
 \le \frac{e}{2\nu} \|\hat{\varrho}\|_{*}^2.
\end{gather*}
Similarly, by using the formula (\ref{I111}) with $a=e^2\nu^2$ and $b=c=|k|^2$, we obtain
\begin{gather*}
\frac{1}{2\pi}\int_{-\infty}^{\infty}\dm s \, 2e^2 s^2 \big(s^2+e^2\nu^2\big)^{-1}
 \big\| \big(s^2+|k|^2\big)^{-1} |k| \hat{\varrho} \big\|_*^2 \\
 \qquad{} = \frac{e^2}{2\pi} \left(\frac{2\pi}{L}\right)^{3}\sum_{k\in M}\frac{2 \pi}{2|k|(|k|+e\nu)^2} |k|^2 |\hat{\varrho}(k)|^2 \le \frac{e}{2\nu} \|\hat{\varrho}\|^2_{*}.
\end{gather*}
Inserting these into (\ref{DefD(s)}), we obtain the assertion (iii).
\end{proof}

\begin{coro}\label{ConvOne}
The right-hand side of \eqref{ESeries} absolutely converges, provided that $\sqrt{2} \| \hat{\varrho}\|_{*}<\nu$, $1\le \sqrt{2}e\nu_0$ and $\sqrt{2} e \|\hat{\varrho}\|_*<1$. In addition,
to exchange the series with the integral in \eqref{ESeries} $($or~\eqref{Expansion}$)$ can be justified.
\end{coro}

\section{Diagonalization II: Two-electron Hamiltonian} \label{Dai2}

Next we will diagonalize $H_{L, \Lambda}$. This is actually possible because we employ the Feynman Hamiltonian, see Remark~\ref{WhyF?} for details. By an argument similar to that of the proof of (\ref{OnePartiK}), $H_{L, \Lambda}$ can be expressed as
\begin{gather*}
H_{L, \Lambda}= -\frac{1}{2}\Delta_1+\frac{1}{2}e^2\nu^2x_1^2
-\frac{1}{2}\Delta_2+\frac{1}{2}e^2\nu^2x_2^2+e \la \mathbb{T}(0)
 {\boldsymbol q}|x_1 \ra_3+e\la\mathbb{T}(r) {\boldsymbol q}| x_2\ra_3\\
\hphantom{H_{L, \Lambda}=}{} +\frac{1}{2}\big(
{\boldsymbol p}^2+\la {\boldsymbol q}| S_0{\boldsymbol q}\ra
\big)-\frac{1}{2}\operatorname{tr}\big[\sqrt{S_0}\big].
\end{gather*}
By setting $\Phi=(x_1, x_2, {\boldsymbol q})$ and $\Pi=(-\im \nabla_1,
-\im \nabla_2, {\boldsymbol p})$, we have that
\begin{gather*}
H_{L, \Lambda}=\frac{1}{2}\big(\la \Pi| \Pi\ra+\la \Phi| \Omega\Phi\ra
\big)-\frac{1}{2}\operatorname{tr}\big[\sqrt{\Omega_0}\big]+3 e\nu,
\end{gather*}
where
\begin{gather*}
\Omega= \Omega_0+Q_1+Q_2,\\ \Omega_0= \begin{pmatrix}
e^2\nu^2 & 0&0\\
0&e^2\nu^2 &0\\
0&0& S_0
\end{pmatrix}, \!\!\qquad Q_1=e \begin{pmatrix}
0&0& \mathbb{T}(0)\\
0&0&0\\
\mathbb{T}^*(0)&0&0
\end{pmatrix},\!\!\qquad
Q_2=e\begin{pmatrix}
0&0&0\\
0&0& \mathbb{T}(r) \\
0&\mathbb{T}^*(r)&0
\end{pmatrix}.
\end{gather*}
By an argument similar to that in the proof of Proposition~\ref{EnergyTr}, we get the following useful formula.
\begin{Prop} Let $E_{L, \Lambda}(R)=\inf \operatorname{spec}(H_{L, \Lambda})$.
If $1\le \sqrt{2} e \nu_0$ and $\sqrt{2}e \|\hat{\varrho}\|_* <1$, then $\Omega\ge 0$ and
\begin{gather*}
E_{L, \Lambda}(R)=\frac{1}{2}\operatorname{tr}\big[\sqrt{\Omega}-\sqrt{\Omega_0}\big]+3e\nu.
\end{gather*}
\end{Prop}
\begin{rem}[Why are the Feynman Hamiltonians helpful?]\label{WhyF?}
From the expression~(\ref{EleFielD}), we see that $\hat{E}_{L, \Lambda}(x)$ can be written as a sum of multiplication operators~$q(k, \lambda)$. As we already knew, this fact is a key to the diagonalization of $H_{L, \Lambda}$. In contrast to the Feynman Hamiltonians, in the standard representation, $\hat{E}_{L, \Lambda}(x)$ corresponds to the following operator:
\begin{gather}
\left(
\frac{2\pi}{L}
\right)^{3/2} \sum_{\lambda=1, 2}\sum_{k\in M}\hat{\varrho}(k)
 \vepsilon(k, \lambda) \big\{\cos(k\cdot x)|k| q(k, \lambda)+\sin(k\cdot x) |k|^{-1}p(k, \lambda)\big\}. \label{QPMix}
\end{gather}
In (\ref{QPMix}), both multiplication and differential operators appear, provided that $x\neq 0$.
At first glance, it appears that diagonalizing the Hamiltonians in this representation requires extra efforts.
\end{rem}

Moreover, it can be readily seen that, by (\ref{I11}),
\begin{gather}
E_{L, \Lambda}(R)
= \frac{1}{2\pi}\sum_{j=1}^{\infty}(-1)^{j+1}\!\int_{-\infty}^{\infty}\!\dm s\,
 s^2 \operatorname{tr}\big[
\big(s^2+\Omega_0\big)^{-1}\big\{
(Q_1+Q_2) \big(s^2+\Omega_0\big)^{-1}
\big\}^j \big]+3e\nu. \!\!\label{TwoElExpansion}
\end{gather}
To examine this formal series, let us introduce the following notation:
\begin{gather*}
\la O_1O_2\cdots O_n\ra
=\frac{1}{2\pi}\int_{-\infty}^{\infty}\dm s\, s^2\operatorname{tr}\big[
\big(s^2+\Omega_0\big)^{-1}O_1(s)O_2(s)\cdots O_n(s) \big],
\end{gather*}
where $O_i(s)=\big(s^2+\Omega_0\big)^{-1/2}O_i \big(s^2+\Omega_0\big)^{-1/2}$.
Then (\ref{TwoElExpansion}) can be expressed as
\begin{gather}
E_{L, \Lambda}(R)
= \sum_{j=1}^{\infty}(-1)^{j+1}\la
 \underbrace{(Q_1+Q_2)\cdots (Q_1+Q_2)}_j\ra+3e\nu\nonumber\\
\hphantom{E_{L, \Lambda}(R)}{}= \sum_{j=1}^{\infty}(-1)^{j+1}\la
 \underbrace{Q_1\cdots Q_1}_j\ra
+\sum_{j=1}^{\infty}(-1)^{j+1}\la
 \underbrace{Q_2\cdots Q_2}_j\ra
\nonumber\\
\hphantom{E_{L, \Lambda}(R)=}{} +\sum_{j=1}^{\infty}\sum_{i_1, \dots, i_j\in \{1,2\}\atop
\{i_1, \dots, i_j\}\neq \{1, 1, \dots, 1\}, \{2,2, \dots, 2\}}
(-1)^{j+1}\la Q_{i_1}\cdots Q_{i_j}\ra+3e\nu. \label{ExpTwo??}
\end{gather}
Since $Q_1$ and $Q_2$ are off-diagonal, we have
\begin{gather*}
E_{L, \Lambda}(R)=
 -\sum_{j=1}^{\infty}\la
 \underbrace{Q_1\cdots Q_1}_{2j}\ra
-\sum_{j=1}^{\infty}\la
 \underbrace{Q_2\cdots Q_2}_{2j}\ra\\
\hphantom{E_{L, \Lambda}(R)=}{} -\sum_{j=1}^{\infty}\sum_{i_1, \dots, i_{2j}\in \{1,2\}\atop
\{i_1, \dots, i_{2j}\}\neq \{1, 1, \dots, 1\}, \{2,2, \dots, 2\}}
\la Q_{i_1}\cdots Q_{i_{2j}}\ra+3e\nu.
\end{gather*}
On the other hand, we remark that, by Corollary~\ref{ConvOne},
\begin{gather*}
E_{L, \Lambda}=
-\sum_{j=1}^{\infty}\la
 \underbrace{Q_1\cdots Q_1}_{2j}\ra+\frac{3}{2}e\nu
=-\sum_{j=1}^{\infty}\la
 \underbrace{Q_2\cdots Q_2}_{2j}\ra+\frac{3}{2}e\nu,
\end{gather*}
provided that $\sqrt{2}\|\hat{\varrho}\|_*<\nu$, $1\le \sqrt{2} e \nu_0$ and $\sqrt{2}e \|\hat{\varrho}\|_* <1$. Thus, we formally arrive at the following formula:
\begin{gather}
E_{L, \Lambda}(R)-2E_{L, \Lambda}
=-\sum_{j=1}^{\infty}\sum_{i_1, \dots, i_{2j}\in \{1,2\}\atop
\{i_1, \dots, i_{2j}\}\neq \{1, 1, \dots, 1\}, \{2,2, \dots, 2\}
}
\la Q_{i_1}\cdots Q_{i_{2j}}\ra.\label{BindingEn}
\end{gather}
Our next task is to prove the convergence of the right-hand side of (\ref{BindingEn}).
For this purpose, we need some preliminaries. Let
\begin{gather*}
\mathcal{I}_{2j}=\big\{I=\{i_1, \dots, i_{2j}\},\, i_1,\dots, i_{2j}\in \{1, 2\}\, \big|\,
I\neq \{1, 1, \dots, 1\}, \{2, 2, \dots, 2\} \big\}.
\end{gather*}
For each $I\in \mathcal{I}_{2j}$, we set $|I|=i_1+i_2+\cdots+i_{2j}$. Furthermore, we use the following notation:
\begin{gather*}
Q_I=Q_{i_1}Q_{i_2} \cdots Q_{i_{2j}},\qquad I=\{i_1, \dots, i_{2j}\}\in \mathcal{I}_{2j}.
\end{gather*}

\begin{lemm}\label{Odd0}
 Let $I\in \mathcal{I}_{2j}$. If $|I|$ is an odd number, then $\la Q_I\ra =0$.
\end{lemm}
\begin{proof}
Note that $\big(s^2+\Omega_0\big)^{-1}$, $Q_1(s)^2$ and $Q_2(s)^2$ are diagonal operators, while $Q_1(s)Q_2(s)$ and $Q_2(s)Q_1(s)$ are off-diagonal operators, see Appendix \ref{List}.
Hence, if $|I|$ is an odd number, then $\big(s^2+\Omega_0\big)^{-1}Q_{i_1}(s)\cdots Q_{i_{2j}}(s)$ is an off-diagonal operator. Accordingly,
\begin{gather*}
\Tr\big[ \big(s^2+\Omega_0\big)^{-1}Q_{i_1}(s)\cdots Q_{i_{2j}}(s)\big]=0.
\end{gather*}
This concludes the proof of Lemma \ref{Odd0}.
 \end{proof}

Let
$\mathcal{I}_{2j}^{(e)}=\{I\in \mathcal{I}_{2j}\, |\, \mbox{$|I|$ is even}\}$. By Lemma~\ref{Odd0}, we have
\begin{gather}
\mbox{the r.h.s.\ of (\ref{BindingEn})}=-\sum_{j=1}^{\infty} \sum_{I\in \mathcal{I}_{2j}^{(e)}}
\la Q_I\ra. \label{ExpansionF}
\end{gather}

\begin{lemm}
For each $s\in \BbbR$ and $I=\{i_1, \dots, i_{2j}\}\in \mathcal{I}_{2j}^{(e)}$, we set
 \begin{gather*}
 Q_I(s)=Q_{i_1}(s)\cdots Q_{i_{2j}}(s)
 \end{gather*}
 and
 \begin{gather*}
 E_I(s)=s^2\operatorname{tr}\big[
 \big(s^2+\omega\big)^{-1} Q_{I\backslash \{i_1\}}^*(s)Q_{I\backslash \{i_1\}}(s)
 \big],
 \end{gather*}
 where $Q_{I\backslash \{i_1\}}^*(s)=\big(Q_{I\backslash \{i_1\}}(s)\big)^*$.
 For all $R>0$, we have the following:
 \begin{itemize}\itemsep=0pt
 \item[{\rm (i)}] For each $s\in \BbbR$ and $I\in \mathcal{I}_{2j}^{(e)}$,
 \begin{gather}
 s^2 \big|\operatorname{tr} \big[
 \big(s^2+\omega\big)^{-1}Q_I(s)
\big]\big| \le D(s)^{1/2} E_I(s)^{1/2}, \label{IDS}
 \end{gather}
 where $D(s)$ is given by \eqref{DefD(s)}.
\item[{\rm (ii)}]
Recall that $a$ is defined by $a=\big(\frac{\sqrt{2}}{\nu}\|\hat{\varrho}\|_*\big)^2$.
If $a<1/4$, then
\begin{gather}
\sum_{j=1}^{\infty}\sum_{I\in \mathcal{I}_{2j}^{(e)}} E_I(s)^{1/2}
\le D(s)^{1/2} \frac{4}{1-4a}. \label{SumEI}
\end{gather}
Thus,
$D(s)^{1/2}\sum\limits_{j=1}^{\infty}\sum\limits_{I\in \mathcal{I}_{2j}^{(e)}} E_I(s)^{1/2}\in L^1(\BbbR)$ and
\begin{gather*}
\sum_{j=1}^{\infty} \sum_{I\in \mathcal{I}_{2j}^{(e)}}
|\la Q_I\ra| \le \frac{e}{\nu} \|\hat{\varrho}\|_*^2\frac{4}{1-4a}.
\end{gather*}
Note that as we mentioned in Lemma~{\rm \ref{EstInt}}, the condition $a<1/4$ is satisfied provided that~$L$ is large enough.
\end{itemize}
\end{lemm}
\begin{proof}For notational simplicity, we set $G=\big(s^2+\omega_0\big)^{-1}$.
By the Schwarz inequality $|\operatorname{tr}[A^*B]| \allowbreak \le \operatorname{tr}[A^*A]^{1/2} \operatorname{tr}[B^*B]^{1/2}$,
we obtain
\begin{gather*}
 \big|\operatorname{tr} \big[
G^{1/2} Q_{i_1}(s)\cdots Q_{i_{2j}}(s) G^{1/2}
\big]\big| \le \operatorname{tr} \big[
G^{1/2}Q_{i_1}(s)Q_{i_1}(s) G^{1/2}
\big]^{1/2} \\
\qquad{}\times \operatorname{tr} \big[
G^{1/2}Q_{i_{2j}}(s) Q_{i_{2j-1}}(s) \cdots Q_{i_2}(s)Q_{i_2}(s) \cdots Q_{i_{2j}}(s)G^{1/2}
\big]^{1/2},
\end{gather*}
which implies that
\begin{gather*}
s^2\big|\operatorname{tr}\big[\big(s^2+\omega\big)^{-1} Q_I(s) \big]\big|\le \big\{s^2 \operatorname{tr}\big[\big(s^2+\omega\big)^{-1} Q_{i_{1}}(s)Q_{i_{1}}(s)\big]\big\}^{1/2} E_I(s)^{1/2}.\label{QSch}
\end{gather*}
Because \begin{gather}
s^2 \operatorname{tr}\big[\big(s^2+\omega\big)^{-1} Q_{i}(s)Q_{i}(s)\big] \le D(s),\label{Q0}
\end{gather}
 we conclude (i).

 From Lemma \ref{EstT} and (\ref{Q^2Formula}), we obtain that
\begin{gather*}
\|Q_i(s)\| \le \frac{\sqrt{2}}{\nu} \|\hat{\varrho}\|_*.\label{Q1}
\end{gather*}
Hence, $E_I(s) \le a^{2j-2}
s^2 \operatorname{tr}\big[\big(s^2+\omega\big)^{-1} Q_{i_{2j}}(s)Q_{i_{2j}}(s)\big]
\le a^{2j-2}D(s)$ by~(\ref{Q0}). Therefore, we obtain~(\ref{SumEI}).

One observes that
\begin{align*}
\sum_{j=1}^{\infty}\sum_{I\in \mathcal{I}_{2j}^{(e)}} s^2\big|\operatorname{tr}\big[\big(s^2+\omega\big)^{-1} Q_I(s) \big]\big|&\underset{(\ref{IDS})}{\le}
\sum_{j=1}^{\infty}\sum_{I\in \mathcal{I}_{2j}^{(e)}} D(s)^{1/2} E_I(s)^{1/2}
 \le \sum_{j=1}^{\infty}2^{2j}a^{j-1}D(s)\\
&\underset{(\ref{SumEI})}{\le} \frac{4}{1-4a}D(s).
 \end{align*}
 In the second inequality, we have used the fact that $\#\mathcal{I}_{2j}^{(e)} \le 2^{2j}$
Accordingly, we get
\begin{gather*}
\sum_{j=1}^{\infty}\sum_{I\in \mathcal{I}_{2j}^{(e)}} |\la Q_I\ra| \le \frac{4}{1-4a} \frac{1}{2\pi} \int_{\BbbR}\dm s \, D(s) \le \frac{4}{1-4a} \frac{e}{\nu}\|\hat{\varrho}\|_*^2
\end{gather*}
by (\ref{EstD}).
\end{proof}

\begin{coro}If $\sqrt{2}\|\hat{\varrho}\|_{*}<\nu$, $1\le \sqrt{2} e \nu_0$ and $\sqrt{2}e \|\hat{\varrho}\|_* <1$, then the r.h.s.\ of \eqref{BindingEn} converges absolutely for every $R>0$. In addition, to exchange the series with the integral, i.e., $\la \cdots \ra$ in~\eqref{BindingEn} $($or \eqref{ExpTwo??}$)$ can be justified.
\end{coro}

\section{Proof of Theorem \ref{Rto7}}\label{PfMainT}

For each $I\in \mathcal{I}_{2j}^{(e)}$, $\# I$ indicates the cardinality of $I$. Notice that $\# I$ is different from $|I|=i_1+\cdots+i_{2j}$.

\subsection[Analysis of $\la Q_I\ra$ with $\#I=2$]{Analysis of $\boldsymbol{\la Q_I\ra}$ with $\boldsymbol{\#I=2}$}

We claim that
\begin{gather}
\la Q_1 Q_2\ra=\la Q_2 Q_1\ra=0. \label{QI=2}
\end{gather}
To see this, let $I=\{1, 2\}$ or $\{2, 1\}$. Trivially, $|I|=1+2=3$. By Lemma~\ref{Odd0}, we conclude~(\ref{QI=2}).

\subsection[Analysis of $\la Q_I\ra$ with $\# I=4$]{Analysis of $\boldsymbol{\la Q_I\ra}$ with $\boldsymbol{\# I=4}$}

In this subsection, we will examine the following terms:
\begin{gather*}
\sum_{i_1, \dots, i_{4}\in \{1,2\}\atop
\{i_1, \dots, i_{4}\}\neq \{1, 1, 1, 1\}, \{2,2,2, 2\}
}\la Q_{i_1}Q_{i_2}Q_{i_3} Q_{i_{4}}\ra =
\mathscr{A}+\mathscr{B},
\end{gather*}
where
\begin{gather}
\mathscr{A}= \la Q_1 Q_1 Q_2 Q_2\ra +\la Q_2Q_2 Q_1 Q_1\ra\label{DefmathA}
\end{gather}
and
\begin{gather}
\mathscr{B}=\la Q_1Q_2Q_1Q_2\ra+\la Q_2 Q_1 Q_2 Q_1\ra+\la Q_2 Q_1 Q_1 Q_2\ra
+\la Q_1 Q_2 Q_2 Q_1\ra.\label{DefmathB}
\end{gather}
In Appendix~\ref{NumC}, we will prove the following lemmas.

\begin{lemm}\label{MainTerm}
We have
\begin{gather*}
\lim_{R\to \infty} \lim_{\Lambda\to \infty} \lim_{L\to \infty} R^7 \mathscr{A} =\frac{23}{4\pi}\left(\frac{1}{2\pi}\right)^2\left(
\frac{1}{4}\alpha_{\mathrm{E, at}}
\right)^2.
\end{gather*}
\end{lemm}

\begin{lemm}\label{Error}We have
\begin{gather*}
\lim_{R\to \infty} \lim_{\Lambda\to \infty} \lim_{L\to \infty} R^9\la Q_1Q_2Q_2Q_1\ra
=\lim_{R\to \infty} \lim_{\Lambda\to \infty} \lim_{L\to \infty} R^9\la Q_2Q_1Q_1Q_2\ra
=\frac{g}{e^2\nu^6},
\end{gather*}
where $g$ is a constant independent of $e$, $\nu_0$ and~$R$. Moreover, $\la Q_1Q_2Q_1Q_2\ra=\la Q_2Q_1Q_2Q_1\ra=0$. Thus, $\lim\limits_{R\to \infty} \lim\limits_{\Lambda\to \infty} \lim\limits_{L\to \infty} R^9\mathscr{B}=2g/e^2\nu^6$.
\end{lemm}

\subsection[Analysis of $\la Q_I\ra$ with $\# I\ge 6$]{Analysis of $\boldsymbol{\la Q_I\ra}$ with $\boldsymbol{\# I\ge 6}$}

Let $I=\{i_1, \dots, i_{2j}\}\in \mathcal{I}_{2j}^{(e)}$.
We will examine the following two cases, separately.
\begin{itemize}\itemsep=0pt
\item[]Case 1: There exists a unique number $\ell\in \{1, 2,\dots, 2j-1\}$ such that $i_{\ell}+i_{\ell+1}=3$.
\item[]Case 2: There exist at least two numbers $m, n\in \{1, 2, \dots, 2j-1\}$ such that
$i_{m}+i_{m+1}=i_{n}+i_{n+1}=3$.
\end{itemize}
\begin{exam}For readers' convenience, we provide some examples below:
\begin{itemize}\itemsep=0pt
\item[] Case 1: $
I=\big\{1,1,\overbrace{1,2}^{i_3+i_4=3},2,2,2, 2\big\}$, $\big\{1,1,1,1,\overbrace{1,2}^{i_5+i_6=3},2,2\big\}$.

\item[] Case 2: $I=\big\{1,1,\overbrace{1,2}^{i_3+i_4=3},2,2,2,\overbrace{2,1}^{i_8+i_9=3},1\big\}$, $\big\{1,1,1,1,1,\overbrace{1,2}^{i_6+i_7=3},\overbrace{2,1}^{i_8+i_9=3},1\big\}$.
\end{itemize}
\end{exam}

\subsubsection{Case 1}

In Appendix \ref{NumC}, we will prove the following lemma.
\begin{lemm}\label{Case1}
Assume that $I$ satisfies the condition in Case~$1$. If $R$ is sufficiently large, then we have
\begin{gather*}
\lim_{\Lambda\to \infty} \lim_{L\to \infty}
|\la Q_I\ra|\le
R^{-9}
\alpha_{\mathrm{E, at}}^2\bigg(
\frac{\|\hat{\varrho}\|^2_{L^2}}{3\nu^2}
\bigg)^{\# I/2-2} C,
\end{gather*}
where $C$ is a positive number independent of $e$, $I$, $R$ and $\nu_0$.
\end{lemm}

\subsubsection{Case 2}
The purpose here is to prove Lemma \ref{QIEst} below. To this end, we begin with the following lemma.

\begin{lemm}\label{EstQG}Let $G=\big(s^2+\Omega_0\big)^{-1}$. For each $j\in \{1, 2\}$, we have
\begin{gather*}
\|Q_j G\| \le D(\hat{\varrho}),
\end{gather*}
where $D(\hat{\varrho})=\max\big\{
\sqrt{2}e \big\||k|^{-1} \hat{\varrho}\big\|_{*}, \frac{\sqrt{2}}{e\nu^2} \||k| \hat{\varrho}\|_{*}\big\}$.
\end{lemm}
\begin{proof} By (\ref{TT*}) and (\ref{T*T}), we readily show that
\begin{gather*}
\big\|\mathbb{T}(r) \big(s^2+S_0\big)^{-1}\big\| \le\sqrt{2}\big\| |k|^{-1} \hat{\varrho}\big\|_{*},\\
\big\|\big(s^2+e^2\nu^2\big)^{-1} \mathbb{T}^*(r)\big\| \le \frac{\sqrt{2}}{e^2\nu^2} \||k| \hat{\varrho}\|_{*}.
\end{gather*}
This concludes the proof of Lemma \ref{EstQG}. \end{proof}

\begin{lemm}\label{QIEst}
Let $j\ge 3$. For each $I\in \mathcal{I}_{2j}^{(e)}$ satisfying the condition in Case~$2$, we have
\begin{gather*}
|\la Q_I\ra| \le c_L^{2j-4} \la Q_2Q_1Q_1Q_2\ra,
\end{gather*}
where
$c_L=\max\big\{D(\hat{\varrho}), \frac{\sqrt{2}} {\nu} \|\hat{\varrho}\|_{*}\big\}$.
\end{lemm}
\begin{proof} By the assumption in the condition Case~2, there exist at least two numbers $m, n\in \{1, 2, \dots, 2j-1\}$ such that $i_m+i_{m+1}=i_n+i_{n+1}=3$. Hence, $I$ can be decomposed as $I=A \cup \{i_m, i_{m+1}\}\cup B \cup \{i_n, i_{n+1}\} \cup C$. Without loss of generality, we may assume that $\{i_m, i_{m+1}\}=\{i_n, i_{n+1}\}=\{1, 2\}$. Thus,
\begin{gather*}
\la Q_I\ra= \la Q_A Q_1Q_2 Q_B Q_1 Q_2 Q_C\ra.
\end{gather*}
Let $Q_I(s)=Q_{i_1}(s)Q_{i_2}(s)\cdots Q_{i_{2j}}(s)$. By the Schwarz inequality, we have
\begin{gather}
\big|\operatorname{tr} \big[G^{1/2} Q_I(s) G^{1/2}\big]\big|\le \Phi_1^{1/2} \Phi_2^{1/2}, \label{Phi}
\end{gather}
where
\begin{gather*}
\Phi_1=\operatorname{tr} \big[G^{1/2}Q_A(s)Q_1(s) Q_2(s) Q_2(s)Q_1(s) Q_A^*(s) G^{1/2}\big],\\
\Phi_2= \operatorname{tr}\big[G^{1/2} Q_C^*(s)Q_2(s)Q_1(s) Q_B^*(s)Q_B(s) Q_1(s)Q_2(s) Q_C(s) G^{1/2}\big].
\end{gather*}
First, we estimate $\Phi_1$. By the cyclic property of the trace, we have
\begin{align}
\Phi_1 &=\operatorname{tr}
\big[
Q_2(s)Q_1(s)Q_A^*(s) G^{1/2} G^{1/2} Q_A(s) Q_1(s)Q_2(s)
\big]\nonumber\\
&= \operatorname{tr}\big[Q_2(s)Q_1(s)G^{1/2} G^{-1/2}Q_A^*(s) G^{1/2} G^{1/2} Q_A(s)G^{-1/2} G^{1/2} Q_1(s)Q_2(s)\big]. \label{Cycle}
\end{align}
Because
\begin{gather*}
G^{1/2} Q_A(s) G^{-1/2}=(GQ_{a_1})(GQ_{a_2}) \cdots (GQ_{a_{\# A}}),
\end{gather*}
where $A=\{a_1, a_2, \dots, a_{\#A}\}$, we have, by Lemma~\ref{EstQG},
\begin{gather*}
\big\|G^{1/2} Q_A(s)G^{-1/2}\big\| \le D(\hat{\varrho})^{\# A}.
\end{gather*}
Thus, by (\ref{Cycle}) and the cyclic property of the trace,
\begin{gather}
\Phi_1 \le D(\hat{\varrho})^{2\# A} \operatorname{tr}
\big[G^{1/2} Q_1(s)Q_2(s)Q_2(s)Q_1(s) G^{1/2}\big]. \label{Cycle2}
\end{gather}
As for $\Phi_2$, we have
\begin{gather*}
\Phi_2 \le \|Q_B(s)\|^2 \operatorname{tr}\big[ G^{1/2} Q_C^*(s) Q_2(s)Q_1(s)Q_1(s)Q_2(s)Q_C(s)G^{1/2} \big].
\end{gather*}
By an argument similar to the one in the proof of (\ref{Cycle2}), one obtains that
\begin{gather}
 \operatorname{tr}
\big[ G^{1/2} Q_C^*(s) Q_2(s)Q_1(s) Q_1(s)Q_2(s)Q_C(s) G^{1/2}
\big] \nonumber\\
\qquad {}\le
D(\hat{\varrho})^{2\# C}
\operatorname{tr} \big[G^{1/2} Q_2(s)Q_1(s)Q_1(s)Q_2(s) G^{1/2}\big]. \label{TrQC}
\end{gather}
By using the fact $\|Q_B(s)\| \le \big(\frac{\sqrt{2}}{\nu} \|\hat{\varrho}\|_{*}\big)^{2\# B}$ and (\ref{TrQC}),
we have
\begin{gather}
\Phi_2 \le \left(
\frac{\sqrt{2}}{\nu}
\| \hat{\varrho}\|_{*}
\right)^{2\# B}
D(\hat{\varrho})^{2\# C}
\operatorname{tr}
\big[
G^{1/2} Q_2(s)Q_1(s)Q_1(s)Q_2(s) G^{1/2}
\big]. \label{Cycle3}
\end{gather}
Combining (\ref{Phi}), (\ref{Cycle2}) and (\ref{Cycle3}), we arrive at
\begin{align*}
|\la Q_I\ra|
& \le \left(
\frac{\sqrt{2}}{\nu}
\|
\hat{\varrho}
\|_{*}
\right)^{\# B}
D(\hat{\varrho})^{\# A+\# C} \la Q_1Q_2Q_2Q_1\ra^{1/2}\la Q_2Q_1Q_1Q_2\ra^{1/2}\\
& \le c_L^{2j-4} \la Q_1Q_2Q_2Q_1\ra^{1/2} \la Q_2Q_1Q_1Q_2\ra^{1/2}.
\end{align*}
Because $\la Q_1Q_2Q_2Q_1\ra=\la Q_2Q_1Q_1Q_2\ra$, we obtain the desired result.
\end{proof}

\subsection{Completion of the proof of Theorem \ref{Rto7}}
First, remark that $\lim\limits_{\Lambda\to \infty} \lim\limits_{L\to \infty}E_{L, \Lambda}=E
$ and $\lim\limits_{\Lambda\to \infty} \lim\limits_{L\to \infty}E_{L, \Lambda}(R)=E(R)
$ by Proposition~\ref{ResolCon}. We divide $\mathcal{I}_{2j}^{(e)}$ as $\mathcal{I}_{2j}^{(e)}=\mathcal{I}^{(e)}_{2j, 1}\cup \mathcal{I}_{2j, 2}^{(e)}$, where
\begin{gather*}
\mathcal{I}_{2j ,\alpha}^{(e)}=\big\{I\in \mathcal{I}_{2j}^{(e)}\, |\, \mbox{$I$ satisfies the condition in Case $\alpha$}\big\},\qquad \alpha=1, 2.
\end{gather*}
Note that $\#\mathcal{I}^{(e)}_{2j, 1}=2j-1$ and $\# \mathcal{I}_{2j, 2}^{(e)}\le 2^{2j}$.

By (\ref{ExpansionF}) and (\ref{QI=2}), one obtains that
\begin{gather*}
2E_{L, \Lambda}-E_{L, \Lambda}(R)=\mathscr{A}+\mathscr{B}+\sum_{j\ge 3}\sum_{I\in \mathcal{I}_{2j, 1}^{(e)}} \la Q_I\ra+\sum_{j\ge 3} \sum_{I\in \mathcal{I}_{2j, 2}^{(e)}} \la Q_I\ra,
\end{gather*}
where $\mathscr{A}$ and $\mathscr{B}$ are defined by (\ref{DefmathA}) and (\ref{DefmathB}), respectively.
Therefore,
\begin{gather}
\big|
R^7
\big\{
2E_{L, \Lambda}-E_{L, \Lambda}(R)-\mathscr{A}
\big\}
\big|
\le
R^7\mathscr{B}+\sum_{j\ge 3} \sum_{I\in \mathcal{I}_{2j, 1}^{(e)}} R^7 |\la Q_I\ra|+\sum_{j\ge 3} \sum_{I\in \mathcal{I}_{2j, 2}^{(e)}} R^7 |\la Q_I\ra|. \label{Last0}
\end{gather}
We will estimate the three terms in the right-hand side of (\ref{Last0}). By Lemma~\ref{Error}, we can easily control the first term. As for the second term, by Lemma~\ref{Case1}, we have
\begin{align}
\lim_{\Lambda\to \infty} \lim_{L\to \infty}\sum_{j\ge 3} \sum_{I\in \mathcal{I}_{2j, 1}^{(e)}} R^7 |\la Q_I\ra|
\underset{{\rm Lemma~\ref{Case1}}}\le & R^{-2} \sum_{j\ge 3}\big(
\# \mathcal{I}^{(e)}_{2j, 1}
\big) C \alpha_{\mathrm{E, at}}^2 \left(
\frac{\|\hat{\varrho}\|_{L^2}^2}{3\nu^2}
\right)^{j-2}\nonumber\\
=\ \ \ \ &R^{-2} C \alpha_{\mathrm{E, at}}^2\sum_{j\ge 3}(2j-1) \left(
\frac{\|\hat{\varrho}\|_{L^2}^2}{3\nu^2}
\right)^{j-2}. \label{Last1}
\end{align}
Note that because $\|\hat{\varrho}\|_{L^2}^2/3\nu^2<1$, the right-hand side of (\ref{Last1}) converges.
On the other hand, using Lemma~\ref{QIEst}, one obtains that
\begin{align}
\sum_{j\ge 3} \sum_{I\in \mathcal{I}_{2j, 2}^{(e)}} R^7 |\la Q_I\ra|
\underset{{\rm Lemma~\ref{QIEst}}}{\le}& \sum_{j\ge 3}
\big(
\# \mathcal{I}_{2j, 2}^{(e)}
\big) c_L^{2j-4} R^7\la Q_2Q_1Q_1Q_2\ra\no
\le\ \ \ \ &
\sum_{j\ge 3}
2^{2j} c_L^{2j-4} R^7\la Q_2Q_1Q_1Q_2\ra. \label{Last2}
\end{align}
Note that because $\lim\limits_{L\to \infty}c_L=c_{\infty}<1/2$, the right-hand side of~(\ref{Last2})
converges, provided that~$L$ is sufficiently large.

Combining (\ref{Last0}), (\ref{Last1}) and (\ref{Last2}), and using Lemma~\ref{Error}, we finally arrive at
\begin{gather*}
 \lim_{R\to \infty}\left|R^7\left\{2E-E(R)-\frac{23}{4\pi}\left(\frac{1}{2\pi}\right)^2\left(
\frac{1}{4}\alpha_{\mathrm{E, at}}\right)^2\right\} \right|\\
\qquad {} \le \left\{2 +c_{\infty}^{-4}\sum_{j\ge 3}(2c_{\infty})^{2j}\right\}
\lim_{R\to \infty } \lim_{\Lambda\to \infty} \lim_{L\to \infty} R^7\la Q_1Q_2Q_2Q_1\ra\\
\qquad\quad{} +\lim_{R\to \infty}R^{-2}C \alpha_{\mathrm{E, at}}^2\sum_{j\ge 3}(2j-1)
\left(\frac{\|\hat{\varrho}\|_{L^2}^2}{3\nu^2}\right)^{2j-2} =0.
\end{gather*}
This concludes the proof of Theorem~\ref{Rto7}.

\section{Discussions}\label{Discuss}
\subsection{Indistinguishability of the electrons}
 The original Hamiltonian $H_{\mathrm{2e}}$ has the indistinguishability of the electrons, i.e.,
the Hamiltonian is unchanged under the exchange of $x_1 \leftrightarrow x_2+r$. In contrast to this,
the approximated Hamiltonian~$H_{\mathrm{D2e}}$ breaks the indistinguishability.
Nevertheless, the Hamiltonian~$H_{\mathrm{D2e}}$ does explain the Casimi--Polder potential as we show in Theorem~\ref{Rto7}. The distinguishability comes from the assumptions (C.1) and (C.2).
 However, to justify the assumptions is still open.

One way to avoid the unjustified derivation of~$H_{\mathrm{D2e}}$ is to
directly start with the Hamiltonian~$H$ given by~(\ref{DefH}) without the last term, which can for
instance be directly taken from \cite[equation~(13.127)]{Spohn} and then extended to the two-particle case. Alternatively and equivalently, the many-particle case is presented, e.g.,
in \cite[Section~4]{Loudon}. If we start from this form, the necessary assumptions are stated as follows:
\begin{itemize}\itemsep=0pt
 \item We assume distinguishability of the two electrons by localizing electron $1$ at $0$, such that electron 1 experiences the field $\hat{E}(0)$, while electron $2$ is localized at~$r$ and hence experiences the field $\hat{E}(r)$.
\item We discard all self-interaction terms and approximate the atomic Coulomb potential by a~harmonic potential.
\end{itemize}
In this manner, we can construct a minimal QED model which describes the Casimir--Polder potential.
Note that, since the particle $1$ and $2$ only communicate via the photon field, and due to distinguishability, the actual choice of coordinate systems is insubstantial such that we can choose for particle~$2$ a coordinate system
that is centered at $r$.

\subsection{Cancellation mechanism of the van der Waals--London force }

As we performed in \cite{MS1}, the attractive $R^{-7}$ decay (the retarded van der Waals potential)
appears due to the exact cancellation of the terms with~$R^{-6}$ decay (the van der Waals--London potential)
originating from $V_R$ by the contribution from the quantized radiation field. Note that the conditions (A.1)--(A.3) are assumed in~\cite{MS1} as well, but (C.1) and (C.2) are not. As we saw in the present paper, this kind of the cancellation mechanism cannot be reproduced under the conditions (C.1), (C.2) and (A.1)--(A.4). In this sense, our assumptions, especially~(C.1) and~(C.2) would be unphysical.

In many literatures, the retardation on the van der Waals potential is examined under the condition (C.1) alone. In these studies, the cancellation of the terms with $R^{-6}$ decay is presupposed and only the $4$-th order perturbation theory is performed without estimating higher order terms.\footnote{A kind of weak cancellation mechanics is discussed in~\cite{Koppen} by imposing the infrared cutoff.} As far as we know, to examine the exact cancellation mechanism under only the condition (C.1) is still unsolved. This problem could be a key to achieving mathematically complete understanding of the retarded van der Waals potential.

\appendix

\section{Useful formulas}\label{List}

In this appendix, we give a list of useful formulas. Let $\mathbb{T}_s(x)= \mathbb{T}(x)\big(s^2+S_0\big)^{-1/2}$. First, we give some formulas for $Q_i$:
\begin{gather}
Q_1(s)=e\big(s^2+e^2\nu^2\big)^{-1/2}\begin{pmatrix}
0&0& \mathbb{T}_s(0)\\
0&0&0\\
\mathbb{T}_s^*(0)&0&0
\end{pmatrix},\nonumber\\
Q_2(s)=e \big(s^2+e^2\nu^2\big)^{-1/2}\begin{pmatrix}
0&0&0\\
0&0& \mathbb{T}_s(r) \\
0&\mathbb{T}_s^*(r)&0
\end{pmatrix},\nonumber\\
Q_1(s)^2=e^2\big(s^2+e^2\nu^2\big)^{-1}\begin{pmatrix}
\mathbb{T}_s(0) \mathbb{T}_s^*(0)&0& 0\\
0&0&0\\
0&0& \mathbb{T}_s^*(0) \mathbb{T}_s(0)
\end{pmatrix},\label{Q3}\\
Q_2(s)^2=e^2 \big(s^2+e^2\nu^2\big)^{-1}\begin{pmatrix}
0&0&0\\
0& \mathbb{T}_s(r)\mathbb{T}_s^*(r)&0 \\
0&0&\mathbb{T}_s^*(r)\mathbb{T}_s(r)
\end{pmatrix}, \label{Q4}\\
Q_1(s)Q_2(s) =e^2\big(s^2+e^2\nu^2\big)^{-1}\begin{pmatrix}
0& \mathbb{T}_s(0)\mathbb{T}_s^*(r)&0 \\
0&0&0\\
0&0& 0
\end{pmatrix},\nonumber\\
Q_2(s)Q_1(s) =e^2 \big(s^2+e^2\nu^2\big)^{-1}\begin{pmatrix}
0&0&0\\
\mathbb{T}_s(r) \mathbb{T}^*_s(0)& 0&0 \\
0&0&0
\end{pmatrix}.\nonumber
\end{gather}

Let $\mathbb{M}(r, r')$ be a linear operator on $\ell^2(M\times \{1, \dots, 4\})$ defined by
\begin{gather*}
(\mathbb{M}(r, r') {\boldsymbol f})(k, \lambda) =\sum_{\lambda'=1}^4\sum_{k'\in M} M_{k, \lambda; k', \lambda'}(r, r') f(k', \lambda'),\qquad {\boldsymbol f} \in \ell^2(M\times \{1, \dots, 4\}),\\ 
M_{k,\lambda; k', \lambda'}(r, r') =\la \vepsilon(k, \lambda)|\vepsilon(k', \lambda')\ra_{3} F_r(k, \lambda)
F_{r'}(k', \lambda').
\end{gather*}

The following formulas are readily checked:
\begin{gather}
\mathbb{T}_s^*(r) \mathbb{T}_s(r') = \big(s^2+S_0\big)^{-1/2} \mathbb{M}(r, r')\big(s^2+S_0\big)^{-1/2}, \label{TT*}\\
\mathbb{T}_s(r) \mathbb{T}_s^*(r') =
\sum_{\lambda=1}^4\sum_{k\in M} |\vepsilon(k, \lambda)\ra\la \vepsilon(k, \lambda)| \big(s^2+k^2\big)^{-1} F_r(k, \lambda)F_{r'}(k, \lambda).\label{T*T}
\end{gather}
Note that $\mathbb{T}_s^*(r) \mathbb{T}_s(r')$ is a map from $\ell^2(M\times \{1, \dots ,4\})$
to $\ell^2(M\times \{1, \dots ,4\})$, while $\mathbb{T}_s(r) \mathbb{T}_s^*(r')$ is a map from $\BbbC^3$ to $\BbbC^3$.

\section{Numerical computations}\label{NumC}

\subsection{Proof of Lemma \ref{MainTerm}}
We will extend the methods in \cite{MS1, MS2}.
By (\ref{Q3}) and (\ref{Q4}), we have
\begin{gather*}
Q_1(s)^2Q_2(s)^2 =e^4 \big(s^2+e^2\nu^2\big)^{-2}\begin{pmatrix}
0&0&0\\
0&0&0 \\
0&0&\mathbb{T}_s^*(0)\mathbb{T}_s(0)\mathbb{T}_s^*(r)\mathbb{T}_s(r)
\end{pmatrix},
\end{gather*}
which implies that
\begin{gather*}
\la Q_1Q_1 Q_2Q_2\ra=\frac{e^4}{2\pi} \int_{\BbbR} \dm s \, \frac{s^2}{\big(s^2+e^2\nu^2\big)^2}
\operatorname{tr}\big[
\mathbb{T}_s^*(0)\mathbb{T}_s(0)\mathbb{T}_s^*(r)\mathbb{T}_s(r)
\big].
\end{gather*}
By (\ref{TT*}) and the fact $\sum\limits_{\lambda_1, \lambda_2=1, 2}(\la \vepsilon(k_1, \lambda_1)|
\vepsilon(k_2, \lambda_2)\ra_3)^2=1+\big(\hat{k}_1\cdot \hat{k}_2\big)^2
$ with $\hat{k}=k/|k|$, we have
\begin{gather}
\operatorname{tr}\big[
\mathbb{T}_s^*(0)\mathbb{T}_s(0)\mathbb{T}_s^*(r)\mathbb{T}_s(r)
\big] =\operatorname{tr}\big[
\big(s^2+\Omega_0\big)^{-2} \mathbb{M}(0, 0)\big(s^2+\Omega_0\big)^{-1}\mathbb{M}(r, r)
\big]\nonumber\\
= \sum_{k_1, \lambda_1}\sum_{k_2, \lambda_2} \big(s^2+k_1^2\big)^{-2}\big(s^2+k_2^2\big)^{-1}
(\la \vepsilon(k_1, \lambda_1)|
\vepsilon(k_2, \lambda_2)\ra_3)^2 \nonumber\\
\quad{} \times k_1^2k_2^2 \hat{\varrho}(k_1)^2\hat{\varrho}(k_2)^2 \cos(k_1\cdot r) \cos(k_2\cdot r)\label{trTTTT}\\
= \sum_{k_1}\sum_{k_2} \big(s^2+k_1^2\big)^{-2}\big(s^2+k_2^2\big)^{-1}
\big(1+\big(\hat{k_1}\cdot \hat{k}_2\big)^2\big) k_1^2k_2^2 \hat{\varrho}(k_1)^2\hat{\varrho}(k_2)^2 \cos(k_1\cdot r) \cos(k_2\cdot r).\nonumber
\end{gather}
By using the assumption (A.3), we have
\begin{gather*}
 \mbox{the r.h.s.\ of (\ref{trTTTT})}\\
\qquad{} = \sum_{k_1}\sum_{k_2} \big(s^2+k_1^2\big)^{-2}\big(s^2+k_2^2\big)^{-1}
\big(1+\big(\hat{k_1}\cdot \hat{k}_2\big)^2\big) k_1^2k_2^2 \hat{\varrho}(k_1)^2\hat{\varrho}(k_2)^2 \cos((k_1+k_2)\cdot r).
\end{gather*}
Hence, we arrive at
\begin{gather*}
\la Q_1Q_1 Q_2Q_2\ra =\frac{e^4}{2}
\left(\frac{2\pi}{L}\right)^{6}\sum_{k_1, k_2\in M}|k_1|^2|k_2|^2 \big(1+\big(\hat{k}_1\cdot \hat{k}_2\big)^2
\big)\hat{\varrho}^2(k_1)\hat{\varrho}^2(k_2) \cos\{(k_1+k_2)\cdot r\} \\
\hphantom{\la Q_1Q_1 Q_2Q_2\ra =}{} \times I_{2,2,1}\big(e^2\nu^2; |k_1|^2; |k_2|^2\big)\label{R7term}
\end{gather*}
where
\begin{gather}
I_{n_a, n_b, n_c}(a; b; c)=
\frac{1}{\pi}\int_{-\infty}^{\infty} \dm s \, \frac{s^2}{\big(s^2+a\big)^{n_a}\big(s^2+b\big)^{n_b}\big(s^2+c\big)^{n_c}}.\label{DefIabc}
\end{gather}
Thus, we obtain that
\begin{gather}
 \lim_{\Lambda\to \infty} \lim_{L\to \infty} \la Q_1Q_1Q_2Q_2\ra
= \frac{e^4}{2} \int_{\BbbR^3\times \BbbR^3} \dm k_1\dm k_2|k_1|^2|k_2|^2 \big(1+\big(\hat{k}_1\cdot \hat{k}_2\big)^2\big)\hat{\varrho}^2(k_1)\hat{\varrho}^2(k_2) \nonumber\\
\hphantom{\lim_{\Lambda\to \infty} \lim_{L\to \infty} \la Q_1Q_1Q_2Q_2\ra =}{}\times \cos\{(k_1+k_2)\cdot r\} I_{2, 2, 1}\big(e^2\nu^2; |k_1|^2; |k_2|^2\big).\label{R7termA}
\end{gather}
By scalings $Rk_1\leadsto k_1$ and $ Rk_2\leadsto k_2$, we have
\begin{gather}
 \mbox{the r.h.s.\ of (\ref{R7termA})} = R^{-10}\frac{e^4}{2}\!
\int\! \dm k_1 \dm k_2\,
|k_1|^2|k_2|^2 \big(1+\big(\hat{k}_1\cdot \hat{k}_2\big)^2\big)\hat{\varrho}^2(k_1/R)\hat{\varrho}^2(k_2/R) \ex^{\im (k_1+k_2)\cdot \hat{n}} \nonumber\\
\hphantom{\mbox{the r.h.s.\ of (\ref{R7termA})} =}{}\times
 I_{2, 2, 1}\big(e^2\nu^2; |k_1|^2/R^2; |k_2|^2/R^2\big), \label{R7term2}
\end{gather}
where $\hat{n}=r/R=(0,0,1)$. Let us switch to spherical coordinates $(r, \vphi, \theta)$ by
\begin{gather*}
\hat{k}=(Y\cos \vphi, Y\sin \vphi, X),\qquad X=\cos \theta, \qquad Y=\sin \theta.
\end{gather*}
 Clearly $X^2+Y^2=1$. Then we have
 \begin{gather*}
 \hat{k}_1\cdot \hat{k}_2=\cos (\vphi_1-\vphi_2)Y_1Y_2+X_1X_2
 \end{gather*}
and hence, by taking the symmetry between $r_1$ and $r_2$ variables into consideration, we obtain
\begin{gather*}
\mbox{the r.h.s.\ of (\ref{R7term2})} =R^{-10} \frac{e^4}{2} \int_0^{\infty} \dm r_1 \int_0^{\infty} \dm r_2 \int_{-1}^1 \dm X_1 \int_{-1}^1 \dm X_2\, \mathfrak{S}(X_1, X_2) \\
\hphantom{\mbox{the r.h.s.\ of (\ref{R7term2})} =}{} \times r_1^4 r_2^4 \ex^{\im r_1X_1} \ex^{\im r_2 X_2} \mathbf{I}\big(e^2\nu^2; r_1^2/R^2; r_2^2/R^2\big)\hat{\varrho}_{\mathrm{rad}}^2 (r_1/R)\hat{\varrho}_{\mathrm{rad}}^2(r_2/R),
\end{gather*}
where
\begin{gather}
\mathfrak{S}(X_1, X_2)= \int_0^{2\pi} \dm \vphi_1 \int_0^{2\pi}\dm \vphi_2\,
\big\{1+\big(\cos(\vphi_1-\vphi_2)Y_1Y_2+X_1X_2\big)^2\big\}\nonumber\\
\hphantom{\mathfrak{S}(X_1, X_2)}{} = 6\pi^2-2\pi^2 \big(X_1^2+X_2^2\big)+6\pi^2 X_1^2X_2^2.\label{DefS}
\end{gather}
and
\begin{gather*}
\mathbf{I}(a; b; c)=\frac{1}{2} \big\{I_{2, 2, 1}(a; b; c)+I_{2, 1, 2}(a; b; c)\big\}.
\end{gather*}
By (\ref{I221}) and (\ref{I212}), we decompose $\mathbf{I}(a; b; c)$ as
\begin{gather*}
\mathbf{I}(a; b; c)=\mathbf{I}_{\mathrm{re}}(a; b; c)+\mathbf{I}_{\mathrm{ir}}(a; b; c),
\end{gather*}
where
\begin{gather*}
\mathbf{I}_{\mathrm{re}}(a; b; c) =\frac{I_{1, 1, 1}(a; b; c)}{8\sqrt{abc}}\left(\frac{1}{A}+\frac{1}{C}\right),\\
\mathbf{I}_{\mathrm{ir}}(a; b; c) =\frac{I_{1, 1, 1}(a; b; c)}{8\sqrt{ab}}\frac{1}{A}\left(
\frac{2}{A}+\frac{1}{C}\right)+ \frac{I_{1, 1, 1}(a; b; c)}{8\sqrt{ac}}\frac{1}{C}\left(
\frac{2}{C}+\frac{1}{A}\right)
\end{gather*}
with $A=\sqrt{a}+\sqrt{b}$, $B=\sqrt{b}+\sqrt{c}$ and $C=\sqrt{c}+\sqrt{a}$. First, we compute the contribution from the term $\mathbf{I}_{\mathrm{re}}$. By the formula
\begin{gather*}
 \mathbf{I}_{\mathrm{re}}\big(e^2\nu^2; r_1^2/R^2; r_2^2/R^2\big)\\
\qquad{} =
\frac{R^3}{8e\nu r_1r_2(e\nu+r_1/R)(e\nu+r_2/R)}
\left(\frac{1}{e\nu+r_1/R}+\frac{1}{e\nu+r_2/R}\right)
 \int_0^{\infty }\dm t \, \ex^{-t (r_1+r_2)},
\end{gather*}
the contribution can be expressed as
\begin{gather}
 \frac{e^3}{16\nu}R^{-7}\int_0^{\infty} \dm t \int_0^{\infty} \dm r_1 \int_0^{\infty}\dm r_2
\left(\frac{1}{e\nu+r_1/R}+\frac{1}{e\nu+r_2/R}\right)\nonumber\\
\qquad{}\times \left\{\frac{3}{2}[1]_{r_1}[1]_{r_2}-[1]_{r_1}[X^2]_{r_2}+\frac{3}{2}[X^2]_{r_1}[X^2]_{r_2}\right\}, \label{IntOsc}
\end{gather}
where
\begin{gather*}
[A(X)]_r(t)=(2\pi)\frac{r^3 \ex^{-tr}\hat{\varrho}_{\mathrm{rad}}^2(r/R)}{e\nu+r/R} \int_{-1}^1 \dm X\, \ex^{\im rX} A(X).
\end{gather*}
For readers' convenience, we will explain how to compute the integral~(\ref{IntOsc}). Let $\vphi(x)$ be the Fourier transformation of $\hat{\varrho}_{\mathrm{rad}}^2(r)$:
$
\vphi(x)=(2\pi)^{-1/2}\int_{\BbbR}\ex^{-\im rx}\hat{\varrho}_{\mathrm{rad}}^2(r)\dm r
$.
Here, we extend $\hat{\varrho}_{\mathrm{rad}}^2$ to a~function on~$\BbbR$ by
$\hat{\varrho}_{\mathrm{rad}}^2(-r) := \hat{\varrho}_{\mathrm{rad}}^2(r)$ for $r>0$.
Note that $\vphi(x)$ decays rapidly by the assumption (A.3). By the convolution theorem in the Fourier analysis, we have
\begin{gather}
 \int_0^{\infty} \dm r\, [1]_r=\int_0^{\infty}\dm s \int_{\BbbR} \dm x \, (1+x/R)\ex^{-se\nu} \frac{
12(t+s/R)^2-4(1+x/R)^2}{\big\{(t+s/R)^2+(1+x/R)^2\big\}^3} \vphi(x) \label{Int1}
\end{gather}
and
\begin{gather*}
\int_0^{\infty} \dm r \, \frac{ [1]_r}{e\nu+r/R}
= \int_0^{\infty}\dm s \int_{\BbbR} \dm x\, (1+x/R)s \ex^{-se\nu} \frac{12(t+s/R)^2-4(1+x/R)^2
}{\big\{(t+s/R)^2+(1+x/R)^2\big\}^3} \vphi(x).
\end{gather*}
[Here, we explain how we derive (\ref{Int1}). First, we observe that
\begin{align*}
\int_0^{\infty}\dm r [1]_r&=4\pi \int_0^{\infty} \dm r\, \frac{r^2\ex^{-tr} \hat{\varrho}_{\mathrm{rad}}^2(r/R)
}{e\nu+r/R} \sin r\no
&= 4\pi \int_0^{\infty}\dm s \, \ex^{-s e\nu } \operatorname{Im} \int_{\BbbR}\dm r \underbrace{1_+(r)r^2 \ex^{-(t+s/R)r}}_{=: f(r)}
\hat{\varrho}_{\mathrm{rad}}^2(r/R)\ex^{\im r},
\end{align*}
where $1_+(r)=1$ if $r>0$, $1_+(r)=0$ if $r\le 0$. By the convolution theorem $\big((2\pi)^{1/2} \big(\hat{g}\hat{h}\big)^{\vee}=g*h\big)$, we have
\begin{gather*}
(2\pi)^{1/2} \operatorname{Im} \int_{\BbbR}\dm r\, f(r)
\hat{\varrho}_{\mathrm{rad}}^2(r/R)\ex^{\im r x}
=\big(\vphi_R* \operatorname{Im} \check{f}\big)(x),
\end{gather*}
where $\vphi_R(x) =R\vphi(Rx)$ and $(g*h)(x)=\int_{\BbbR} g(y)h(x-y)\dm y$.
Because
\begin{gather*}
 \operatorname{Im} \check{f}(x)=(2\pi)^{-1/2} \frac{6x(t+s/R)^2-2x^3}{\big\{(t+s/R)^2+x^2\big\}^3},
\end{gather*}
we get (\ref{Int1}).] Hence, by the dominated convergence theorem, we obtain
\begin{gather*}
 \lim_{R\to \infty}\frac{e^3}{16\nu}\int_0^{\infty} \dm t \int_0^{\infty} \dm r_1 \int_0^{\infty}\dm r_2
\left(\frac{1}{e\nu+r_1/R}+\frac{1}{e\nu+r_2/R}\right) [1]_{r_1}[1]_{r_2}\\
\qquad{} = \frac{\pi}{8\nu^4} \hat{\varrho}^4(0)\int_0^{\infty} \dm t \, \frac{\big({-}4+12t^2\big)^2}{\big(t^2+1\big)^6}.
\end{gather*}
Similarly, we obtain that
\begin{gather*}
 \lim_{R\to \infty}\frac{e^3}{16\nu}\int_0^{\infty} \dm t \int_0^{\infty} \dm r_1 \int_0^{\infty}\dm r_2
\left(
\frac{1}{e\nu+r_1/R}+\frac{1}{e\nu+r_2/R}
\right) [1]_{r_1}\big[X^2\big]_{r_2}\\
\qquad{} = \frac{\pi}{8\nu^4} \hat{\varrho}^4(0)\int_0^{\infty} \dm t \, \frac{\big({-}4+12t^2\big)\big({-}12+4t^2\big)}{\big(t^2+1\big)^6}
\end{gather*}
and
\begin{gather*}
 \lim_{R\to \infty}\frac{e^3}{16\nu}\int_0^{\infty} \dm t \int_0^{\infty} \dm r_1 \int_0^{\infty}\dm r_2
\left(\frac{1}{e\nu+r_1/R}+\frac{1}{e\nu+r_2/R}\right) \big[X^2\big]_{r_1}\big[X^2\big]_{r_2}\\
\qquad{} = \frac{\pi}{8\nu^4} \hat{\varrho}^4(0)\int_0^{\infty} \dm t\, \frac{\big({-}12+4t^2\big)^2}{\big(t^2+1\big)^6}.\label{Int6}
\end{gather*}
Summarizing the above results, we arrive at
\begin{gather*}
\frac{e^4}{2}
\int_0^{\infty} \dm r_1 \int_0^{\infty} \dm r_2 \int_{-1}^1 \dm X_1 \int_{-1}^1 \dm X_2 \, \mathfrak{S}(X_1, X_2) r_1^4 r_2^4 \ex^{\im r_1X_1} \ex^{\im r_2 X_2} \\
\qquad{} \times \mathbf{I}_{\mathrm{re}}\big(e^2\nu^2; r_1^2/R^2; r_2^2/R^2\big)\hat{\varrho}^2_{\mathrm{rad}}(r_1/R)\hat{\varrho}^2_{\mathrm{rad}}(r_2/R)
= \frac{23\pi^3}{2\nu^4}\hat{\varrho}^4(0) R^{-7}+o\big(R^{-7}\big),
\end{gather*}
where we used the following formula in \cite{MS1}:
\begin{gather*}
\int_0^{\infty}\dm t \left\{
\frac{3}{2}A(t)^2-A(t)B(t)+\frac{3}{2}B(t)^2
\right\}=92\pi^3
\end{gather*}
with
\[
A(t)=\frac{-4+12t^2}{\big(1+t^2\big)^3} \qquad \text{and}\qquad
B(t)=\frac{4\big({-}3+t^2\big)}{\big(1+t^2\big)^3}.
\]

As for the contribution from $\mathbf{I}_{\mathrm{ir}}$, we have, by an argument similar to that of the computation concerning with $\mathbf{I}_{\mathrm{re}}$,
\begin{gather*}
 \frac{e^4}{2}
\int_0^{\infty} \dm r_1 \int_0^{\infty} \dm r_2 \int_{-1}^1 \dm X_1 \int_{-1}^1 \dm X_2\, \mathfrak{S}(X_1, X_2) \ r_1^4 r_2^4 \ex^{\im r_1X_1} \ex^{\im r_2 X_2} \\
\qquad {} \times \mathbf{I}_{\mathrm{ir}}\big(e^2\nu^2; r_1^2/R^2; r_2^2/R^2\big)\hat{\varrho}_{\mathrm{rad}}^2(r_1/R)\hat{\varrho}_{\mathrm{rad}}^2(r_2/R)
= \mathrm{const} \cdot R^{-9}+o\big(R^{-9}\big).
\end{gather*}
To summarize, we obtain that
\[
\lim_{R\to \infty} \lim_{\Lambda\to \infty} \lim_{L\to \infty} R^7\la Q_1Q_1Q_2Q_2\ra= \frac{23}{8\pi}\left(\frac{1}{2\pi}\right)^2\left(
\frac{1}{4}\alpha_{\mathrm{E, at}}
\right)^2.
\]
Similarly, we get
\[
\lim_{R\to \infty} \lim_{\Lambda\to \infty} \lim_{L\to \infty} R^7\la Q_2Q_2Q_1Q_1\ra= \frac{23}{8\pi}\left(\frac{1}{2\pi}\right)^2\left(
\frac{1}{4}\alpha_{\mathrm{E, at}}\right)^2.
\]
This concludes the proof of Lemma~\ref{MainTerm}.

\subsection{Proof of Lemma \ref{Error}}

We readily see that $\la Q_1Q_2Q_1Q_2\ra=\la Q_2Q_1Q_2Q_1\ra=0$ by the formulas in Appendix \ref{List}.
In what follows, we evaluate $\la Q_1Q_2Q_2Q_1\ra$ and $\la Q_2Q_1Q_1Q_2\ra$.
Because the argument here is almost pallarel to the proof of Lemma~\ref{MainTerm}, we provide a sketch only.
As before, we have
\begin{gather}
 \lim_{\Lambda\to \infty}\lim_{L\to \infty}\la Q_2Q_1Q_1Q_2\ra
=\frac{e^4}{2\pi} \int\dm k_1 \dm k_2\,
\big( 1+\big(\hat{k}_1\cdot \hat{k}_2\big)^2 \big)
|k_1|^2|k_2|^2 \hat\varrho^2(k_1) \hat{\varrho}^2(k_2) \nonumber\\
\qquad\quad{} \times \cos(k_1\cdot r) \cos (k_2\cdot r)
I_{3, 1,1}\big(e^2\nu^2; |k_1|^2; |k_2|^2\big)\nonumber\\
\qquad{} = R^{-10} \frac{e^4}{2}
\int_0^{\infty} \dm r_1 \int_0^{\infty} \dm r_2 \int_{-1}^1 \dm X_1 \int_{-1}^1 \dm X_2 \, \mathfrak{S}(X_1, X_2) r_1^4 r_2^4 \hat{\varrho}_{\mathrm{rad}}^2(r_1/R)\hat{\varrho}_{\mathrm{rad}}^2(r_2/R) \nonumber\\
\qquad\quad{} \times \cos (r_1X_1)\cos(r_2 X_2)I_{3, 1, 1}\big(e^2\nu^2; r_1^2/R^2; r^2_2/R^2\big).\label{Q2112}
\end{gather}
Remark the following formula:
\begin{gather*}
I_{3, 1, 1}\big(e^2\nu^2; r_1^2/R^2; r^2_2/R^2\big)= (e\nu)^{-6} \frac{R}{r_1+r_2}+o(R),
\end{gather*}
which follows from (\ref{I311}).
Inserting this into (\ref{Q2112}), we formally obtain that
\begin{gather*}
\lim_{\Lambda\to \infty}\lim_{L\to \infty}\la Q_2Q_1Q_1Q_2\ra=\frac{g}{e^2\nu^6}R^{-9}+o\big(R^{-9}\big).
\end{gather*}
To justify this rough argument, we carefully have to treat the oscillatory integral as we did in the proof of Lemma \ref{MainTerm}. Similarly, we see that $\lim\limits_{\Lambda\to \infty}\lim\limits_{L\to \infty}\la Q_1Q_2Q_2Q_1\ra=\frac{g}{e^2\nu^6}R^{-9}+o\big(R^{-9}\big).$

\subsection{Proof of Lemma \ref{Case1}}
In this case, there exist two numbers $m, n\in \BbbN$ with $m+n\ge 3$ such that $\la Q_I\ra=\big\la Q_1^{2m} Q_2^{2n}\big\ra$ or $\la Q_I\ra=\big\la Q_2^{2m}Q_1^{2n}\big\ra$. We will study the case where $\la Q_I\ra=\big\la Q_1^{2m}Q_2^{2n}\big\ra$ only. By using the formulas in Appendix~\ref{List}, one obtains that
\begin{gather}
 \lim_{\Lambda\to \infty}\lim_{L\to \infty} \big\la Q_1^{2m}Q_2^{2n}\big\ra \nonumber\\
 \qquad{}
= e^{2(m+n) }\sum_{\lambda_1, \dots, \lambda_{m+n}=1, 2} \int\dm k_1\cdots \dm k_{m+n}\,
I_{m+n, 2, 1^{m+n-1}}\big(e^2\nu^2; |k_1|^2; \dots; |k_{m+n}|^2\big) \nonumber\\
\qquad\quad{}\times \left[\prod_{j=1}^{m+n} |k_j|^2\hat{\varrho}^2(k_j)\right]
\la \vepsilon_1|\vepsilon_2\ra\la \vepsilon_2|\vepsilon_3\ra\cdots \la \vepsilon_{m+n}|\vepsilon_1\ra
\cos(k_m\cdot r) \cos(k_{m+1}\cdot r), \label{QQ1}
\end{gather}
where $\vepsilon_j=\vepsilon(k_j, \lambda_j)$ and
$
I_{m+n, 2, 1^{m+n-1}}(a_0; \dots ;a_{m+n})=I_{m+n, 2, \underbrace{1, \dots, 1}_{m+n-1}}(a_0; \dots; a_{m+n})
$ with
\begin{gather*}
I_{n_0,n_1, \dots, n_k}(a_0;a_1; \dots; a_k)
=\frac{1}{\pi}\int_{-\infty}^{\infty} \dm s\, \frac{s^2}{\prod\limits_{j=0}^k(s^2+a_j)^{n_j}}.
\end{gather*}
By scalings $Rk_m\leadsto k_m$ and $Rk_{m+1}\leadsto k_{m+1}$, we get
\begin{gather}
 \mbox{the r.h.s.\ of (\ref{QQ1})}= e^{2(m+n) } R^{-10}
\sum_{\lambda_1, \dots, \lambda_{m+n}=1, 2} \int\dm k_1\cdots \dm k_{m+n} \nonumber\\
\quad {}\times I_{m+n, 2, 1^{m+n-1}}\big(
e^2\nu^2; |k_1|^2; \dots ;|k_{m-1}|^2; |k_m|^2/R^2; |k_{m+1}|^2/R^2; |k_{m+2}|^2; \dots; |k_{m+n}|^2
\big) \nonumber\\
\quad{}\times \bigg[
\prod_{j\neq m, m+1} |k_j|^2 \hat{\varrho}^2(k_j)
\bigg]\la \vepsilon_1|\vepsilon_2\ra\la \vepsilon_2|\vepsilon_3\ra\cdots \la \vepsilon_{m+n}|\vepsilon_1\ra \nonumber\\
\quad{}\times |k_m|^2|k_{m+1}|^2 \hat{\varrho}^2(k_m/R) \hat{\varrho}^2(k_{m+1}/R)
\cos(k_m\cdot \hat{n}) \cos(k_{m+1}\cdot \hat{n}).\label{QQ2}
\end{gather}
Switching to the polar coordinates as we did in the proof of Lemma \ref{MainTerm}, we see that
\begin{gather}
 \mbox{the r.h.s.\ of (\ref{QQ2})}=e^{2(m+n)}R^{-9}\sum_{\lambda_1, \dots, \lambda_{m+n}=1, 2} \left[
\prod_{j=1}^{m+n} \int \dm r_j \int \dm X_j \int \dm \vphi_j\right]
\left[
\frac{1}{\pi} \int_{-\infty}^{\infty} \dm s\,s^2
\right] \nonumber\\
\qquad{} \times
\mathscr{F}_R(r_1, \dots, r_{m-1}, r_{m+2}, \dots, r_{m+n}; s) \mathscr{G}_R(r_m, r_{m+1}; s) \nonumber\\
\qquad{} \times \la \vepsilon_1|\vepsilon_2\ra\la \vepsilon_2|\vepsilon_3\ra\cdots \la \vepsilon_{m+n}|\vepsilon_1\ra \cos(r_mX_m) \cos(r_{m+1}X_{m+1}), \label{QQQ1}
\end{gather}
where
\begin{gather*}
\mathscr{F}_R =\big(e^2\nu^2+R^{-2} s^2\big)^{-(m+n)}
\Bigg[\prod_{j\neq m, m+1}r_j^4 \hat{\varrho}_{\mathrm{rad}}^2(r_j) \big(r_j^2+R^{-2}s^2\big)^{-1}\Bigg],\\
\mathscr{G}_R=r_m^4r_{m+1}^4 \hat{\varrho}_{\mathrm{rad}}^2(r_m/R)\hat{\varrho}_{\mathrm{rad}}^2(r_{m+1}/R)
\big(s^2+r_m^2\big)^{-1}\big(s^2+r_{m+1}^2\big)^{-1}.
\end{gather*}
Next, we will perform $X_j$- and $\vphi_j$-integrations for $j\neq m, m+1$.
For this purpose, we remark that
\[
\sum_{\lambda=1, 2} \int_{-1}^1 \dm X \int_0^{2\pi} \dm \vphi\, |\vepsilon(k, \lambda)\ra\la \vepsilon(k, \lambda)|=\frac{4\pi}{3} \one_3,
\]
where $\one_3$ is the identity matrix acting in $\BbbC^3$.
Using this and the fact that
\[
\sum_{\lambda_m, \lambda_{m+1}=1, 2} (\la \vepsilon_m|\vepsilon_{m+1}\ra_3)^2
=1+\big(\hat{k}_{m}\cdot \hat{k}_{m+1}\big)^2,
\]
we get
\begin{gather}
 \mbox{the r.h.s.\ of (\ref{QQQ1})} =e^{2(m+n)}R^{-9}\left(
\frac{4\pi}{3}
\right)^{m+n-2}\left[
\frac{1}{\pi} \int_{-\infty}^{\infty} \dm ss^2
\right]
\left[
\prod_{j\neq m, m+1}\int \dm r_j
\right] \mathscr{F}_R \nonumber\\
\quad{}\times
\left[
\int \dm r_m \dm r_{m+1} \int \dm X_m \dm X_{m+1}
\right]\mathfrak{S}(X_m, X_{m+1}) \mathscr{G}_R \cos(r_mX_m) \cos(r_{m+1}X_{m+1}), \label{QQQ2}
\end{gather}
where $\mathfrak{S}(X_m, X_{m+1})$ is defined by~(\ref{DefS}). Because
\begin{gather*}
\mathscr{F}_R\le \big(e^2\nu^2\big)^{-(m+n)}
\Bigg[\prod_{j\neq m, m+1}r_j^2\hat{\varrho}^2_{\mathrm{rad}}(r_j)\Bigg],
\end{gather*}
we obtain that
\begin{gather}
 \mbox{the r.h.s.\ of (\ref{QQQ2})} \le R^{-9} \left(\frac{\|\hat{\varrho}\|_{L^2}^2}{3\nu^2}
\right)^{m+n-2}\nu^{-4}\frac{1}{\pi}\int_{-\infty}^{\infty}\dm s \, s^2
 \int \dm r_m\dm r_{m+1} \nonumber\\
 \hphantom{\mbox{the r.h.s.\ of (\ref{QQQ2})} \le}{} \times \!\int\!\dm X_m\dm X_{m+1}\,
\mathscr{G}_R \mathfrak{S}(X_m, X_{m+1}) \cos(r_mX_m) \cos (r_{m+1}X_{m+1}).\!\!\!\!\!\label{QQQ3}
\end{gather}
Here, we used the fact that the factor
\[
\int\dm r_m \int\dm r_{m+1}\int \dm X_m \dm X_{m+1}[\cdots]
\] in the r.h.s.\ of~(\ref{QQQ3}) is positive for all $s\ge 0$. Using the elementary formula
\begin{gather*}
\frac{1}{\pi} \int_{-\infty}^{\infty} \dm s \, \frac{s^2}{\big(s^2+r_m^2\big)\big(s^2+r_{m+1}^2\big)}=\frac{1}{r_m+r_{m+1}}=\int_0^{\infty} \dm t\, \ex^{-t(r_m+r_{m+1})},
\end{gather*}
we have
\begin{gather}
\mbox{the r.h.s.\ of (\ref{QQQ3})}=R^{-9}
\left(\frac{\|\hat{\varrho}\|^2_{L^2}}{3\nu^2}
\right)^{m+n-2} \nonumber\\
\hphantom{\mbox{the r.h.s.\ of (\ref{QQQ3})}=}{}\times \nu^{-4} \int_0^{\infty}\dm t \left\{
\frac{3}{2} [\![1]\!] [\![1]\!] - [\![1]\!]\big[\!\big[X^2\big]\!\big]+\frac{3}{2}\big[\!\big[X^2\big]\!\big]\big[\!\big[X^2\big]\!\big]\right\}, \label{QQ6}
\end{gather}
where
\[
[\![A(X)]\!](t)= (2\pi)\int_0^{\infty} \dm r\, r^4 \ex^{-tr} \hat{\varrho}_{\mathrm{rad}}^2(r/R) \int_{-1}^1 \dm X \ex^{\im rX} A(X).
\]
We can compute $[\![1]\!]$ and $[\![X^2]\!]$ as
\begin{gather*}
[\![1]\!] =\int_{\BbbR} \dm x\, \frac{-24t^3+24t(1+x/R)^2}{\big\{t^2+(1+x/R)^2 \big\}^4} \vphi(x),\\
[\![X^2]\!]=\int_{\BbbR}\dm x\, F_t(1+x/R)\vphi(x),
\end{gather*}
where
\begin{gather*}
F_t(a)=8\frac{(1-a)t^5-2a\big(a^2+a-3\big)t^3-a^3\big(a^2+3a+6\big)t}{\big(t^2+a^2\big)^4}.
\end{gather*}
Since $\vphi(x)$ decays rapidly, we readily see that the integral in~(\ref{QQ6}) is uniformly bounded provided that~$R$ is sufficiently large.

\section[Basic properties of $I_{n_a, n_b, n_c}(a; b; c)$]{Basic properties of $\boldsymbol{I_{n_a, n_b, n_c}(a; b; c)}$}\label{BasicI}

Here, we will give a list of basic properties of $I_{n_a, n_b, n_c}(a; b; c)$ defined by~(\ref{DefIabc}).

The following result is easily checked:
\begin{gather}
I_{1, 1, 1}(a; b; c)=\frac{1}{ABC}, \label{I111}
\end{gather}
where $A=\sqrt{a}+\sqrt{b}$, $B=\sqrt{b}+\sqrt{c}$ and $C=\sqrt{c}+\sqrt{a}$. Using this, we have
\begin{gather}
I_{2, 2, 1}(a; b; c) =\frac{1}{4\sqrt{ab}} I_{1, 1, 1}(a; b; c) \left\{
\frac{2}{A^2} +\frac{1}{AC}+\frac{1}{AB}+\frac{1}{BC}\right\},\label{I221}\\
I_{2, 1, 2}(a; b; c) =\frac{1}{4\sqrt{ac}} I_{1, 1, 1}(a; b; c) \left\{
\frac{2}{C^2}+\frac{1}{AC}+\frac{1}{BC}+\frac{1}{AB}\right\}\label{I212}
\end{gather}
and
\begin{gather}
I_{3, 1, 1}(a; b; c) =\frac{1}{8a}I_{1, 1, 1}(a; b; c)\left\{
\frac{2}{A^2}+\frac{2}{C^2}+\frac{2}{AC}+\frac{1}{\sqrt{a}A}+\frac{1}{\sqrt{a}C}\right\}.\label{I311}
\end{gather}

\subsection*{Acknowledgements}
The original idea of the present paper comes from an unpublished sketch by Herbert Spohn. I~would like to thank the kind referees for very helpful comments. The discussions in Section~\ref{Discuss} heavily rely on their comments. This work was partially supported by KAKENHI 18K03315.

\pdfbookmark[1]{References}{ref}
\LastPageEnding

\end{document}